\documentclass[12pt]{article}
\usepackage{graphicx, multirow}
\usepackage{multirow}
\usepackage{subfigure}
\usepackage{natbib}
\usepackage{threeparttable}
\usepackage{epsfig}
\usepackage{amssymb,amsmath,amsfonts,latexsym,amsthm}
\usepackage{hyperref}
\hypersetup{colorlinks=true,linkcolor=blue,citecolor=blue}
\usepackage{bm}% bold math
\usepackage[usenames, table, dvipsnames]{xcolor}

\newcommand{\beq}{\begin{equation}}
\newcommand{\eeq}{\end{equation}}
\newcommand{\beqs}{\begin{eqnarray}}
\newcommand{\eeqs}{\end{eqnarray}}

\newcommand{\rz}{\mathbb{R}}

\newcommand{\cof}{\textup{cof}}

\newcommand{\mathscr}{\mathcal}
\newcommand{\bfa}{{\bf a}}
\newcommand{\bfb}{{\bf b}}
\newcommand{\bfc}{{\bf c}}
\newcommand{\bfd}{{\bf d}}
\newcommand{\bfe}{{\bf e}}

\newcommand{\bfg}{{\bf g}}

\newcommand{\bfm}{{\bf m}}
\newcommand{\bfn}{{\bf n}}

\newcommand{\bfs}{{\bf s}}

\newcommand{\bfu}{{\bf u}}
\newcommand{\bfv}{{\bf v}}

\newcommand{\bfx}{{\bf x}}
\newcommand{\bfy}{{\bf y}}

\newcommand{\bfz}{{\bf z}}
\newcommand{\bfA}{{\bf A}}
\newcommand{\bfB}{{\bf B}}

\newcommand{\bfE}{{\bf E}}
\newcommand{\bfF}{{\bf F}}

\newcommand{\bfH}{{\bf H}}
\newcommand{\bfI}{{\bf I}}
\newcommand{\bfQ}{{\bf Q}}
\newcommand{\bfR}{{\bf R}}
\newcommand{\bfS}{{\bf S}}

\newcommand{\bfU}{{\bf U}}

\newcommand{\bfW}{{\bf W}}

\newcommand{\vphi}{{\varphi}}
\newcommand{\eps}{{\varepsilon}}

\newcommand{\beql}{\begin{equation} \label}

\newcommand{\calP}{{\cal P}}

\newtheorem{theorem}{Theorem}

\newtheorem{lem}[theorem]{Lemma}
\newtheorem{proposition}[theorem]{Proposition}
\newtheorem{cor}[theorem]{Corollary}
\theoremstyle{definition}

\theoremstyle{remark}

\begin{document}

 \title{Study of the {\it cofactor conditions}: conditions of supercompatibility
between phases}

\author{\normalsize Xian Chen$^{a}$, Vijay Srivastava$^{b}$, Vivekanand Dabade$^{a}$, and Richard D. James$^{a}$ \\
\normalsize\it$^{a}$ Department of Aerospace Engineering and Mechanics, University of Minnesota, \\
\normalsize\it Minneapolis, Minnesota 55455, USA\\
\normalsize\it$^{b}$ GE Global Research Center, Niskayuna, New York 12309, USA
}
\date{\normalsize \today}
 
% \begin{center}
%{\LARGE Study of the {\it cofactor conditions}: conditions of supercompatibility
%between phases}
%\vspace{5mm}
%{Xian Chen}, {xian@aem.umn.edu}  \\
%{Vijay Srivastava}, {srivasta@ge.com}  \\
%{Vivekanand Dabade}, {dabad001@umn.edu} \\
%{Richard D. James}, {james@umn.edu}  \\
%\end{center}

%\author[label1]{Xian Chen}
% \ead{chen1561@umn.edu}
% \ead[url]{www.tc.umn.edu/~chen1561}
 %\author[label2]{Vijay Srivastava}
 %\author[label1]{Vivekanand Dabade}
 %\author[label1]{Richard D. James \corref{To whom all correspondences should be addressed.}}
 %\ead{james@aem.umn.edu}
 %\address[label1]{110 Union St. SE, Minneapolis, MN 55455 USA}
 %\address[label2]{GE Global Research Center, 1 Research Circle, MB 274A Niskayuna, NY 12309 USA}

\maketitle

\begin{abstract}
The {\it cofactor conditions},  introduced in \citet{james_2005}, are
conditions
of compatibility between
 phases in martensitic materials.  They consist of three subconditions:
i) the condition that the middle
principal stretch of the transformation stretch tensor $\bfU$ is unity
($\lambda_2 = 1$),  ii) the condition $\bfa \cdot
\bfU\, \cof (\bfU^2 - \bfI)\bfn  = 0$, where the vectors $\bfa$ and $\bfn$
are certain vectors arising in the specification of the twin system,
and iii) the inequality ${\rm tr} \bfU^2 + \det \bfU^2 -
(1/4) |\bfa|^2 |\bfn|^2
\ge 2$.  Together, these conditions are necessary
and sufficient for the equations of the
crystallographic theory of martensite to be satisfied for the given
twin system but for any volume fraction $f$ of the twins, $0 \le f \le 1$.
This contrasts sharply with the generic solutions  of the crystallographic
theory which have at most two such volume fractions for a given twin system
of the form $f^*$ and $1-f^*$.
In this paper we simplify the form of the cofactor conditions, we
give their specific forms for various symmetries and twin types,
we clarify the
extent to which the satisfaction of the cofactor conditions for one
twin system implies its satisfaction for other twin systems.  In particular,
we prove that the satisfaction of the cofactor conditions for either
Type I or Type II twins implies that there are solutions of the
crystallographic theory using these twins that have {\it no elastic transition
layer}.  We show that the latter further implies macroscopically
curved, transition-layer-free
austenite/martensite interfaces for Type I twins, and planar
transition-layer-free
interfaces for Type II twins which nevertheless
permit significant flexibility
(many deformations) of the martensite.  We identify some real material systems
nearly satisfying the cofactor conditions.
Overall, the cofactor conditions are shown to dramatically increase
the number of deformations possible in austenite/martensite mixtures without the
presence of elastic energy needed for coexistence.  In the context of
earlier work that links the special case $\lambda_2 = 1$ to reversibility
\citep{jun_06, zhiyong_08, zarnetta_08},
it is expected that satisfaction of the cofactor conditions for Type I
or Type II twins will
lead to further lowered hysteresis and improved resistance to transformational
fatigue in alloys whose composition has been tuned to satisfy
these conditions.
\end{abstract}

%\begin{keyword}
%Martensitic Phase Transformation \sep Compatibility \sep
%Hysteresis \sep {\it Cofactor Conditions} \sep Triple Junctions \sep Microstructures
%\end{keyword}
\newpage
\tableofcontents

% \linenumbers


\section{Introduction}
\label{intro}

This paper gives a precise derivation and
implications of the cofactor conditions
\citep{james_2005},
defined briefly in the abstract.
% and more precisely below.
These conditions are appropriate to a material that undergoes an austenite to
martensitic phase transformation having
symmetry-related variants of martensite.
The cofactor conditions represent a degeneracy of the equations
of the crystallographic theory of martensite  \citep{lieberman_55, bowles_54a, bowles_54b},
%\footnote{Here we
%use the analytical treatment of the crystallographic theory given in
%\citep{ball_james_87}}
under which this theory possesses
solutions with any volume fraction $0 \le f \le 1$ of the
twins \citep{james_2005}.

For the special cases $f = 0$ and $f = 1$ the equations of the
crystallographic theory reduce to the equations of compatibility between
austenite and the appropriate single variant of martensite.  Hence,
as also can be seen from the conditions themselves (in particular,
the condition $\lambda_2 = 1$),
the cofactor
conditions imply perfect compatibility between austenite
and each single variant
of martensite.  The solutions of the crystallographic
theory for the intermediate volume fractions
$0 < f < 1$  imply the existence of the
standard low energy transition layers between
austenite and finely-twinned martensite. 

The main result of this paper is that in many cases, the cofactor
conditions imply that the transition layer can be eliminated altogether, resulting in the coexistence
of austenite and twinned martensite with zero elastic energy.  Examples are shown in Figures \ref{curve_inf} (right), \ref{typeI_interface}, \ref{type2_interface}, \ref{nucleation_AinM} and \ref{nucleation_MinA}. These include macroscopically curved austenite/martensite interfaces and natural mechanisms of nucleation.
The latter are continuous families of deformations in which the austenite grows from zero volume in a matrix of martensite, or the martensite grows in a matrix of austenite, all having zero elastic energy.  Said differently, while the crystallographic theory implies that the energy due to elastic distortion can be reduced as close to zero as desired by making the twins finer and finer, the elastic energy in the cases studied here is eliminated at all length scales.  From a physical viewpoint, the only remaining energy is then a small interfacial energy.  We describe explicitly the cases in which the transition layer can be eliminated in Section \ref{cof:condn}.
%What was not realized
%in previous papers \citep{zhiyong_08, jun_06, james_2005},
%and is shown in detail here, is that, in many cases, the cofactor
%conditions imply that the transition
%layer can be eliminated altogether, resulting in the coexistence of
%austenite and twinned martensite with zero elastic energy (See Figures
%1-3).

%We show that the cofactor conditions also imply the existence of
%natural mechanisms of nucleation.  These are continuous families of
%deformations in which the austenite grows in a matrix of martensite,
%or the martensite grows in a matrix of austenite, both having zero
%elastic energy.

The value of $\lambda_2$ can be modified
by changing composition, and the special case
$\lambda_2 = 1$ (up to experimental error in
the measurement of lattice parameters) has been achieved in many systems.
As reviewed in detail below, satisfaction of only the condition
$\lambda_2 = 1$
has a dramatic effect on hysteresis and transformational fatigue
\citep{jun_06, zhiyong_08, zarnetta_08, remi_09,vijay_10}; see also
\citep{buschbeck_11, meethong_07, louie_10, vijay_11}.
A theory for the width of the hysteresis loop that predicts this
sensitivity was given in \citep{zhiyong_08, knupfer_2011, zwicknagl_2013}.  It is based on the
idea that transformation
is delayed, say on cooling, because the additional bulk and twin-boundary
energy at the austenite/martensite interface has to be compensated by
a further lowering of the energy wells of the martensite phase, so as
to have a free energy decreasing transformation path.
This bulk and interfacial
energy is eliminated by tuning composition to make $\lambda_2 = 1$.
Both this theory
and broad collection of measurements of hysteresis demonstrate extreme
sensitivity of the width of the hysteresis to $\lambda_2$ (and composition),
which also explains why this was not observed previously.  For example, as
shown in Figure \ref{hires}, 1/4 \% changes of composition
in the Ti$_{50}$Ni$_{50-x}$Pd$_{x}$ system give
a minimum width   of the  hysteresis loop  at $x = 9.25$ with a remarkable value
$(1/2) (A_f + A_s - M_f- M_s) =\ $2$^{\circ}$C.  This is accompanied
by improvements of the reversibility of the phase transformation as measured by
the migration of the transformation temperature under repeated cycling.

\begin{figure}[htp]
\centering
\includegraphics[width=0.6 \textwidth]{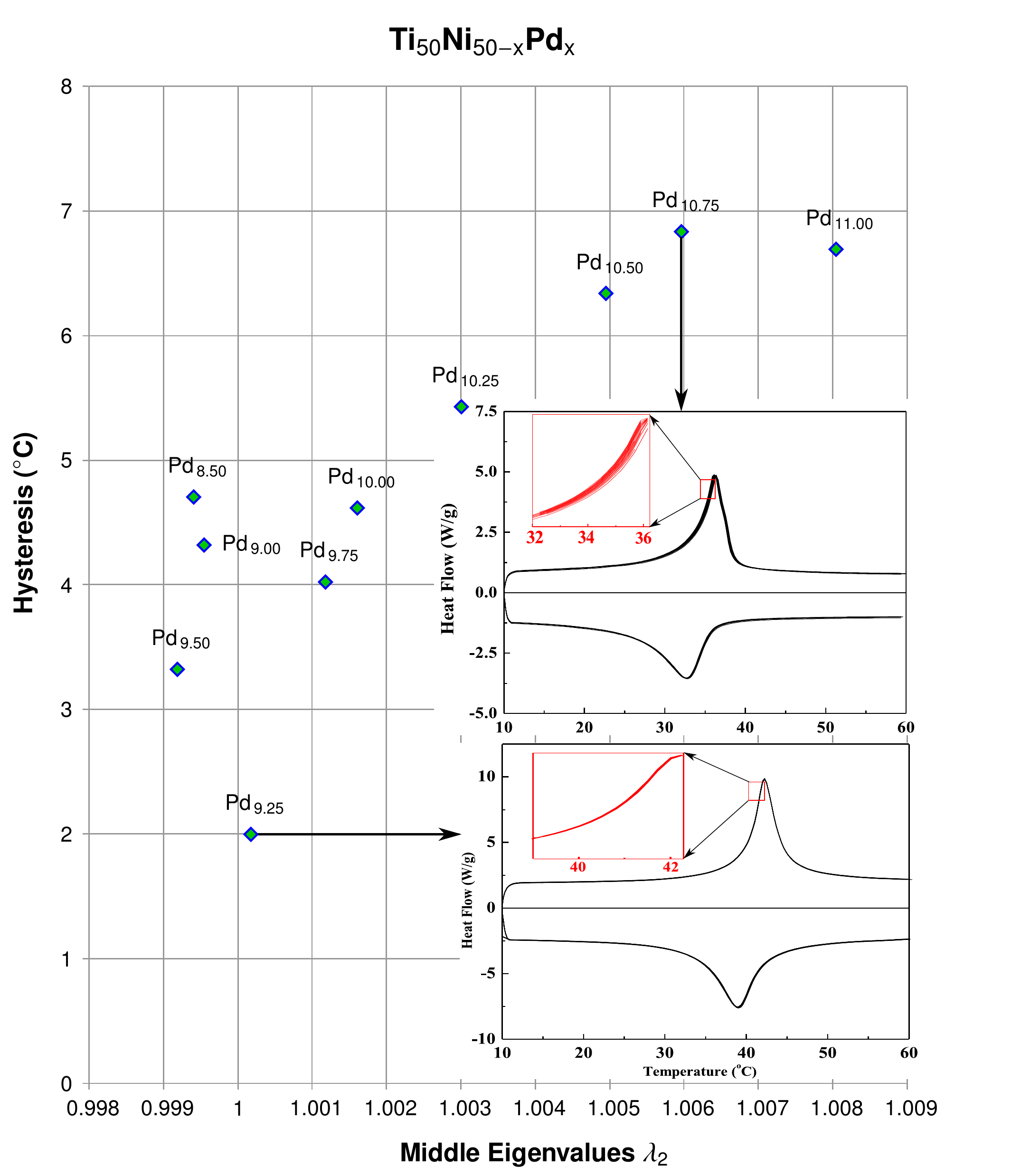}
\caption{Reduced hysteresis of Ti$_{50}$Ni$_{50-x}$Pd$_x$ alloy
system as the composition is tuned to achieve $\lambda_2 = 1$.  A
thermal hysteresis of 2$^\circ$ C is obtained at $x = 9.25$.  The insets
show a comparison of thermal hysteresis under repeated cycling
through the transformation (30 cycles) measured by differential
scanning calorimetry at $x = 9.25$ vs.~$x = 10.75$. A careful comparison
of these graphs shows an average migration of transformation temperature
of $0.16^\circ$C/cycle at $x = 10.75$ is reduced to
$0.030^\circ$C/cycle at $x = 9.25$.  These values should be contrasted
to ordinary
TiNi which exhibits an average migration over 30 cycles of about
$0.6^\circ$C/cycle.   \label{hires}}
\end{figure}

%Under the cofactor conditions, not only is the elastic transition layer
%eliminated in austenite/single variant martensite interfaces, but it is
%also removed in a variety of microstructures of austenite coexisting with
%martensite.
%
%

Tuning $\lambda_2$ to $1$ actually entails  a reduction of the
number of deformations that belong to solutions of the crystallographic
theory in many cases.  This can be seen in the following way.
In general, for $\lambda_2$ near 1 but $\lambda_2 \ne 1$,
the crystallographic
theory implies the existence of four solutions per
twin system \citep{ball_james_87}, resulting in four average
deformation gradients of twinned
laminates that participate in austenite/martensite interfaces.
As $\lambda_2 \to 1$, these four solutions converge to four perfect
austenite/single-variant martensite interfaces.  (This is consistent
with the fact that when the middle eigenvalue $\lambda_2$ of a
positive-definite symmetric tensor $\bfU$ is 1, there
are two solutions $\bfR_1, \bfa_1 \otimes \bfn_1$   and
$\bfR_2, \bfa_2 \otimes \bfn_2$  of the equation of perfect compatibility
$\bfR \bfU - \bfI = \bfa \otimes \bfn$,
$\bfR \in$ SO(3), $\bfa, \bfn \in \rz^3$ \citep{ball_1992}.)
However, some of these four also result from other solutions of the
crystallographic theory, because a variant can belong to many
twin systems.  In fact, a simple counting exercise shows that
the number of deformation gradients participating in exact
interfaces equals the number of generic twin
systems \citep{pitteri_98}.  For example, in a classic cubic to orthorhombic
phase
transformation \citep{zhiyong_08} as in the material TiNiPd
(Figure \ref{hires}),
there are 6 variants of martensite, resulting generically
in 30 twin systems and 24 (resp., 96) solutions of the crystallographic
theory for $\lambda_2 \lesssim 1$ (resp., $\lambda_2 \gtrsim 1$).
If $\lambda_2 = 1$ in this case, there are only 30
deformation gradients corresponding to exact austenite/martensite
interfaces. 

Fewer deformation gradients means fewer ways that nontransforming impurities,
defects, triple junctions and precipitates can be accommodated by a growing
austenite/martensite interface.  This intuition on the beneficial effects
of having more deformations, which is prevalent
in the literature on phase transformations, is quantified in random
polycrystals by \citep{bhatta_96}.  This line of thought
also plays an important role in the concept of non-generic twins of
\citet{pitteri_98}.  
As summarized above, if the cofactor conditions are satisfied, there are infinitely many deformation gradients participating in austenite/twinned-martensite interfaces. As mentioned above, in some cases (Type I or Type II but generally not Compound twins, see below) the elastic transition layer can be eliminated. Particularly in these cases, the demonstrated advantages with regard to hysteresis and reversibility of having no transition layer are combined with the benefits of having a great many deformations. The precise nature of these possible benefits with regard to the shape memory effect or transformational fatigue awaits further theoretical and experimental study.
%If the cofactor conditions
%are satisfied, there are  infinitely many deformation gradients
%participating in austenite/twinned-martensite interfaces.  In some cases
%(Type I or Type II but not Compound twins, see below) the elastic
%transition layer can be eliminated.  Particularly in these cases, the demonstrated advantages with
%regard to hysteresis and reversibility of having
%no transition layer are combined with the benefits of having a great
%many deformations.  The precise nature of these possible benefits with
%regard to the shape memory effect or transformational fatigue await
%further theoretical and experimental study.

This paper unifies the treatment of compatibility of variants of martensite,
by including automatically Type I/II and Compound twins,
the ``domains'' of Li \citep{jian_95, jian_97},
and the non-conventional and non-generic twins of Soligo, Pitteri and
Zanzotto \citep{pitteri_98, soligo_99}.
All of these cases can satisfy the cofactor conditions,
and all of these cases are analyzed here.

Geometrically linear theory is often used in the literature.
We present a treatment of the cofactor conditions in that
case.  They can be obtained either by direct linearization of
the cofactor conditions of the geometrically nonlinear theory, or by
starting over and imposing the condition of ``any volume fraction of the
twins'' in the geometrically linear form of the crystallographic theory.

\vspace{2mm}
\noindent {\it Notation and method of visualization}:
As general background for this paper, including most
notation, see the book of \citet{kaushik_03}.
$\mathbb R^{3 \times 3}$ is the set of $3 \times 3$
real matrices, $\mathbb R^{3 \times 3}_+$ is the subset
of $\mathbb R^{3 \times 3}$ with positive determinant,
$\mathbb R^{3 \times 3}_\text{+sym}$ is positive-definite, symmetric
real $3 \times 3$ matrices,
and $\text{SO}(3)$ denotes the group of all $3 \times 3$ orthogonal
matrices with determinant 1. The notation
$\textup{cof}\, \mathbf A$ denotes the cofactor of the matrix
$\mathbf A$: in components relative to an orthonormal basis,
$(\textup{cof} \mathbf A)_{ij} = (-1)^{i+j}\textup{det}
(\hat{\mathbf A}_{ij})$,
where $\hat{\mathbf A}_{ij}$ is the  determinant of the submatrix
obtained by removing the $i$th row
and $j$th column of $\bfA$.  The pictures of microstructures shown in
this paper are plotted using the following algorithm: a) A deformation
$\bfy(\bfx)$ defined on a cube $\Omega$ and having the given
values of $\nabla \bfy$, e.g., those arising from materials satisfying the
cofactor conditions, is constructed analytically\footnote{In cases that this
deformation contains a transition layer at an interface, linear
interpolation of the
deformation across this layer is used, unless otherwise noted.}.
b) Suitable rectangular arrays of points $\bfx_1, \bfx_2, \dots \in \partial \Omega$
are specified.
c) Dots at the points $\bfy(\bfx_1), \bfy(\bfx_2), \dots$ are plotted, colored
by their phase or variant.
This is a direct visualization via the Cauchy-Born
rule.
%Each circle in Figure \ref{circ_diagram} near a matrix $\bfU_i$
%denotes the set SO(3)$\bfU_i = \{\bfR \bfU_i: \bfR \in \text{SO}(3) \}$.
%The symbols M/M and A/M are shorthand
%for martensite/martensite and austenite/martensite interfaces,
%respectively.


\section{Geometrically nonlinear theory of martensite and the crystallographic
theory} \label{geomart}

The cofactor conditions arise as degeneracy conditions in the
crystallographic theory of martensite, but they have wider implications
for the existence of energy minimizing microstructures within the geometrically nonlinear
theory of martensitic transformations.  Thus we present a brief summary of
the parts of the theory that are needed in this paper.
As general references we cite
\citet{kaushik_03, james_2000, ball_james_87}.

The domain $\Omega \subset \mathbb R^3$, interpreted as a region occupied by
undistorted
austenite at the transformation temperature,  serves as reference
configuration for deformations
$\mathbf{y} : \Omega \rightarrow \rz^3$ arising from transformation or elastic
distortion.
The total energy of an unloaded body  subjected to a deformation
$\mathbf{y} : \Omega \rightarrow \rz^3$ at a temperature
$\theta$ is given by
\begin{equation}
\int_\Omega \vphi(\nabla \mathbf{y}(\bfx), \theta)\,
d \mathbf{x} \label{tot_energy}.
\end{equation}
The Helmholtz free energy per unit reference volume,
$\vphi(\bfF, \theta)$, depends on the deformation
gradient $\bfF \in \mathbb R^{3 \times 3}_+$ and
the absolute temperature $\theta >0$.  This energy density can be related
to  atomistic
theory by the Cauchy-Born rule \citep{pitteri_03}.  In this scenario
$\bfF$ is interpreted
as a linear transformation locally mapping a Bravais lattice representing
undistorted
austenite to the martensite lattice.  If the austenite is
represented by a complex lattice consisting of the union of several
Bravais lattices, all having the same lattice vectors but having different
base points $\bfa_1, \dots, \bfa_m$,  the appropriate version of the Cauchy-Born
rule -- the {\it weak Cauchy-Born rule} in the terminology of
\citet{pitteri_98} and \citet{ericksen_08}  --  gives an energy density of the form
$\hat{\vphi}(\bfF, \bfa_m - \bfa_1, \dots, \bfa_2 - \bfa_1, \theta)$.
In that case the free energy density given above is defined by
\beq
  \vphi (\bfF, \theta) = \min_{\bfs_1, \dots, \bfs_{m-1}}
  \hat{\vphi}(\bfF, \bfs_1, \dots, \bfs_{m-1}, \theta).
\eeq
 The free energy density $\vphi$ is frame-indifferent,
$\vphi(\bfR \bfF, \theta) = \vphi(\bfF, \theta)$ for all $\theta>0,\ \bfR \in
\text{SO(3)}$ and $\bfF \in \mathbb R^{3 \times 3}_\text{+}$,
and its energy-well structure is restricted by conditions of
symmetry which are not repeated here.

The result is that there is a set of {\it transformation stretch matrices}
$\bfU_1, \dots, \bfU_n$, each in $\mathbb R^{3 \times 3}_\text{+sym}$,
that are related by symmetry, $\bfU_i = \bfQ_i \bfU_1 \bfQ_i^T$,
$i = 1, \dots, n$, where $\calP  = \{\bfQ_1, \dots, \bfQ_n\}, \bfQ_i \in\ $O(3) is the point group
of undistorted austenite at $\theta_c$.  $\bfU_1, \dots, \bfU_n$ define the energy
wells of the {\it variants of martensite}.
That is, there is a {\it transformation temperature} $\theta_c$
such that
\beq
  \vphi(\bfU_1, \theta) = \dots = \vphi(\bfU_n, \theta) \le
  \vphi(\bfF, \theta), \quad \theta \le \theta_c.
\eeq
The matrices $\bfU_i = \bfQ_i \bfU_1 \bfQ_i^T$,
$i = 1, \dots, n$ depend weakly on temperature, due to ordinary
thermal expansion, but this dependence is suppressed.

For $\theta = \theta_c$, the identity $\bfI$, representing the austenite,
 is also a minimizer:
\beq
 0 = \vphi(\bfI, \theta_c) = \vphi(\bfU_1, \theta_c) \le
 \vphi(\bfF, \theta_c).
\eeq
Without loss of generality we have put the minimum value of the
energy at $\theta_c$ equal to zero.
As $\theta$ is increased from $\theta_c$ the austenite well persists,
but it is perturbed slightly away from $\bfI$ due again to ordinary
thermal expansion.  $\bfU_1, \dots, \bfU_n$ also can be continued
as local minimizers of the energy density  for $\theta> \theta_c$.
 While there are various obvious generalizations of
our results, in this paper we nominally discuss energy
minimizers and minimizing
sequences at $\theta_c$.  In summary, the full set of
minimizers of the free energy density $\vphi$ at $\theta_c$ includes
\beq
   \text {SO(3)}\bfI\ \cup\   \text {SO(3)}\bfU_1\ \cup\ \cdots
   \cup\  \text {SO(3)}\bfU_n
\eeq
for given symmetry-related tensors $\bfU_1, \dots, \bfU_n$
in $\mathbb R^{3 \times 3}_\text{+sym}$.  To avoid degeneracy
we assume that $\bfI, \bfU_1, \dots, \bfU_n$ are distinct.  A
general algorithm that can be used to obtain the transformation stretch
matrices directly from x-ray measurements, applicable also to
complex lattices,  is presented in a forthcoming paper
 \citep{chen_12a}.
 
\subsection{Twins and domains}
\label{twins_domains}

Accounting for frame-indifference, the equation of compatibility for
two variants of martensite is
\beq
\hat{\mathbf{R}} \mathbf{U}_i - \bar{\mathbf{R}}\mathbf{U}_j = \mathbf{a}
\otimes
\mathbf{n}, \label{twin_rel}
\eeq
which is to be solved for $\hat{\bfR}, \bar{\bfR} \in\ $SO(3) and
$\bfa, \bfn \in \mathbb R^3$.   Without loss of generality,
we can put $ \bar{\bfR} = \bfI$
and $j=1$.  The former is accomplished by premultiplying (\ref{twin_rel})
by $\bar{\bfR}^T$ (corresponding to an overall rigid rotation)
and suitably redefining $\hat{\bfR}$ and $\bfa$.
The latter is
accomplished by subsequently pre- and post- multiplying (\ref{twin_rel}) by
$\bfQ_j, \dots, \bfQ_j^T$ and using the symmetry relations above.
Thus we consider
\beq
\hat{\mathbf{R}} \mathbf{U}_i - \mathbf{U}_1 = \mathbf{a} \otimes
\mathbf{n}. \label{twin_rel1}
\eeq
To recover the general case (\ref{twin_rel}) we multiply (\ref{twin_rel1})
by $\bfQ_j, \dots, \bfQ_j^T$  and then premultiply by an arbitrary
$\bar{\bfR} \in\ $SO(3) and make the obvious notational changes.

Because of results given in the Appendix and described in the
following paragraphs, it is seen that the details
of symmetry relations, the number of variants, point groups, etc., do
not play a direct role in the analysis.  So we simplify the notation.
  Let
$\bfU = \bfU_1 \in \mathbb R^{3 \times 3}_\text{+sym}$ and
$\hat{\bfU} \in \mathbb R^{3 \times 3}_\text{+sym}$.
Let $\hat{\bfR} \in {\rm SO(3)}, \bfa, \bfn \in \rz^3$ satisfy
\beq
  \hat{\bfR} \hat{\bfU} - \bfU =  \bfa \otimes \bfn.  \label{twin_rel2}
\eeq

It is known that the solutions of the equation of compatibility
(\ref{twin_rel2}) between
martensite variants can be classified into five types: Type I,
Type II, Compound, non-conventional but generic and non-generic twins.
The terminology non-generic twins and non-conventional twins was introduced
by Pitteri and Zanzotto \citep{pitteri_98, soligo_99} in the
context of cubic to monoclinic transformations.  Briefly,
Type I/II twins are the well-known solutions generated by a two-fold
$\bfQ \in \calP$
such that $\bfU_j = \bfQ \bfU_1 \bfQ^T \ne \bfU_1$.  Compound
twins are possible when there are two distinct two-fold transformations
relating $\bfU_j$ and $\bfU_1$
and can be considered as both Type I and Type II simultaneously.
Non-conventional twins are solutions of (\ref{twin_rel2}) that are not
generated by a two-fold transformation in $\calP$ but that persist under
arbitrary small perturbations of $\bfU_1$, and non-generic twins
are solutions of (\ref{twin_rel2}) that do not persist
under arbitrary small perturbations of $\bfU_1$ and therefore can be considered
as associated to special choices of the lattice parameters.  Both non-generic
and non-conventional twins do not in general have a mirror symmetry relation
across the interface.   Or, more precisely, if atom positions on each
side of interface are constructed using the Cauchy-Born rule
and non-generic or non-conventional
solutions of (\ref{twin_rel2}),
then generally there will be no mirror symmetry relating the
atom positions across the interface.  Noticing this fact from a
purely experimental viewpoint in LaNbO$_4$, Li
referred to these structures as ``domains'' rather
than twins in his thesis \citep{jian_95}.

In the Appendix we show that all solutions of
(\ref{twin_rel2}) can be expressed in a common form by simple formulas.
In particular, these formulas include Types I/II, Compound, non-conventional
and non-generic twins,
as well as cases that may occur with other symmetries that have not yet been
classified.
Our analysis of the cofactor conditions below relies only on the presence of
these formulas, so we use this framework below.  Our formulas have the
same form as for Type I/II twins with an associated two-fold rotation
(which is given by an explicit formula), but this two-fold rotation is
not generally in $\calP$.  For this reason we here use the terminology
of Li and call these general solutions  {\it Type I domains} and
{\it Type II domains} (see also the case of {\it Compound domains}
defined below).
It can be seen from the Appendix that these domains are
twins with respect to a mythical
symmetry, not the symmetry of lattices of austenite and martensite
consistent with the framework above.

The analysis, under the hypotheses on $\bfU, \hat{\bfU}$ given
above, that all solutions of (\ref{twin_rel2}) (and therefore of
(\ref{twin_rel}))
are Type I, Type II or Compound domains is given in the Appendix.
The proposition
given there implies that if $\hat{\bfR}, \bfa, \bfn$ satisfy
(\ref{twin_rel2}), then there is a unit vector $\hat{\bfe}$
such that
\beq
 \hat{\bfU} =  (-\bfI + 2 \hat{\bfe} \otimes \hat{\bfe})\bfU
  (-\bfI + 2 \hat{\bfe} \otimes \hat{\bfe}),   \label{ehati}
\eeq
and it therefore follows by standard results (see \citet{kaushik_03})
 that there are two
solutions $(\bfR_I, \bfa_I \otimes \bfn_I)$ and
$(\bfR_{II}, \bfa_{II} \otimes \bfn_{II})$
of (\ref{twin_rel2}) given by
\beq
\begin{array}{lll}
\text{Type I } ~&\mathbf{n}_{I} = \hat{\mathbf{e}}, ~ &\mathbf{a}_I =
2 (\dfrac{\mathbf{U}^{-1} \hat{\mathbf{e}}}{|\mathbf{U}^{-1}
\hat{\mathbf{e}}|^2} - \mathbf{U} \hat{\mathbf{e}}),  \\
\text{Type II } ~ &\mathbf{n}_{II} = {2}(\hat{\mathbf{e}} -
\dfrac{\mathbf{U}^2 \hat{\mathbf{e}}}{|\mathbf{U} \hat{\mathbf{e}}|^2}),
 ~ &\mathbf{a}_{II} =  \mathbf{U} \hat{\mathbf{e}}.
\end{array} \label{typeI_twin}
\eeq
Following this specification of $\bfa_I \otimes \bfn_I$ and
$\bfa_{II} \otimes \bfn_{II}$, the corresponding
rotations $\bfR_I$ and $\bfR_{II}$ can be calculated from
(\ref{twin_rel2}).  Note that by changing $\bfa \to \rho\, \bfa$
and $\bfn \to (1/\rho) \bfn$, $\rho \ne 0$, we do not change
$\bfa \otimes \bfn$, so these individual vectors are not uniquely determined
by the solution.  This situation occurs widely below, and so
statements about uniqueness or numbers of solutions always refer to
the diadic $\bfa \otimes \bfn$ rather than the individual vectors.
This observation can be used to normalize $\bfn$, up to
$\pm$, but we do not do that in this paper.

As seen from Corollary \ref{cd} of the Appendix, there are cases
in which $\bfU$ and $\hat{\bfU}$ are related as in (\ref{ehati}) by
{\it two} nonparallel unit vectors $\hat{\bfe}_1, \hat{\bfe}_2$.
This apparently gives rise to four solutions of (\ref{twin_rel2})
via (\ref{typeI_twin}), but these solutions cannot be distinct due to the
fact that there are at most two solutions
$\hat{\bfR}, \bfa \otimes \bfn$  of (\ref{twin_rel2})
according to Prop. 4 of  \citet{ball_james_87}.  One solution
can be considered Type I for $\hat{\bfe}_1$ and
Type II for $\hat{\bfe}_2$ and the other is
Type II for $\hat{\bfe}_1$ and Type I for $\hat{\bfe}_2$.
In the conventional cases of
twins, these degenerate solutions are interpreted as Compound twins.
Corollary \ref{cd}
and (\ref{typeI_twin}) show that the same situation can arise in the general
case of the Appendix.  Therefore we use  the following terminology
throughout the rest of this paper.
We call the solutions given in (\ref{typeI_twin})  {\it Type I/II domains}
in the case that there is one and only one unit vector $\hat{\bfe}$
satisfying (\ref{ehati}) (up to $\pm$) and
$\bfa_I \otimes \bfn_I/\bfa_{II} \otimes \bfn_{II}$ is given
by the first line/second line  of  (\ref{typeI_twin}). In cases where there
are two nonparallel unit vectors satisfying (\ref{ehati}), we call the
resulting pair of solutions {\it Compound domains}.

Compound domains are characterized below.
\begin{proposition}  \label{lem1} (Compound domains) Assume
that $\bfU \in \mathbb R^{3\times3}_{\rm +sym}$.
Let $|\hat{\bfe}_1| = 1$ be given,  define
$\hat{\bfU} = (-\bfI + 2 \hat{\bfe}_1 \otimes \hat{\bfe}_1)\bfU
(-\bfI + 2 \hat{\bfe}_1 \otimes \hat{\bfe}_1)$ and suppose
$\hat{\bfU} \ne \bfU$.
There is a second unit vector $\hat{\bfe}_2$, not parallel to $\hat{\bfe}_1$,
satisfying $\hat{\bfU} = (-\bfI + 2 \hat{\bfe}_2 \otimes \hat{\bfe}_2)\bfU
(-\bfI + 2 \hat{\bfe}_2 \otimes \hat{\bfe}_2)$  if and only if
$\hat{\bfe}_1$ is perpendicular to an eigenvector of $\bfU$.
In the case that $\hat{\bfe}_1$ is perpendicular to an eigenvector of $\bfU$,
$\hat{\bfe}_2$ is unique up to $\pm$ and is
perpendicular to both $\hat{\bfe}_1$ and that eigenvector.

Supposing that $\hat{\bfe}_1$ is perpendicular to
an eigenvector $|\bfv| = 1$ of $\bfU\ (\ne \hat{\bfU})$
and $\hat{\bfe}_2  = \bfv \times \hat{\bfe}_1$, then the two solutions
$\bfa_C^1 \otimes \bfn_C^1,\ \bfa_C^2 \otimes \bfn_C^2$ of (\ref{twin_rel2})
can be written
\beqs
 \bfn_C^1 &=& \hat{\bfe}_1, \quad  \bfa_C^1 = \xi \bfU \hat{\bfe}_2,
 \quad \quad  {\rm where} \ \
\xi = 2 \frac{\hat{\bfe}_2 \cdot \bfU^{-2} \hat{\bfe}_1}
{\hat{\bfe}_1 \cdot \bfU^{-2} \hat{\bfe}_1}, \nonumber \\
 \bfn_C^2 &=& \hat{\bfe}_2, \quad  \bfa_C^2 = \eta \bfU \hat{\bfe}_1,
 \quad \quad  {\rm where} \ \
\eta = -2 \frac{\hat{\bfe}_2 \cdot \bfU^{2} \hat{\bfe}_1}
{\hat{\bfe}_1 \cdot \bfU^{2} \hat{\bfe}_1}.
\label{comp}
\eeqs
 \end{proposition}
\begin{proof} Suppose $\hat{\bfe}_1 \cdot \bfv = 0$ for some $|\bfv| = 1$
satisfying $\bfU \bfv = \bfv$.  Define $\hat{\bfe}_2 = \hat{\bfe}_1 \times \bfv$
so that $\hat{\bfe}_1, \hat{\bfe}_2, \bfv = 0$ is an orthonormal basis.
Then, $(-\bfI + 2 \hat{\bfe}_1 \otimes \hat{\bfe}_1)
(-\bfI + 2 \hat{\bfe}_2 \otimes \hat{\bfe}_2) = -\bfI + 2 \bfv \otimes \bfv$.
Since $(-\bfI + 2 \bfv \otimes \bfv) \bfU (-\bfI + 2 \bfv \otimes \bfv) = \bfU$,
we have
\beq
(-\bfI + 2 \hat{\bfe}_2 \otimes \hat{\bfe}_2)\bfU
(-\bfI + 2 \hat{\bfe}_2 \otimes \hat{\bfe}_2)  =
(-\bfI + 2 \hat{\bfe}_1 \otimes \hat{\bfe}_1)\bfU
(-\bfI + 2 \hat{\bfe}_1 \otimes \hat{\bfe}_1).  \label{two}
\eeq

Conversely, if there are two nonparallel unit vectors
$\hat{\bfe}_1, \hat{\bfe}_2$
satisfying (\ref{two}), then
by Corollary \ref{cd} of the Appendix, $\hat{\bfe}_1 \cdot \hat{\bfe}_2 = 0$.
Let $\bfv =  \hat{\bfe}_1 \times \hat{\bfe}_2 $, so that $|\bfv|  = 1$
and $(-\bfI + 2 \hat{\bfe}_1 \otimes \hat{\bfe}_1)
(-\bfI + 2 \hat{\bfe}_2 \otimes \hat{\bfe}_2) = -\bfI + 2 \bfv \otimes \bfv$.
Hence it follows from (\ref{two}) that
$(-\bfI + 2 \bfv \otimes \bfv) \bfU (-\bfI + 2 \bfv \otimes \bfv) = \bfU$.
Operating the latter on $\bfv$  it is seen that $\bfv$ is an eigenvector
of $\bfU$, so $\hat{\bfe}_1$ is perpendicular to an eigenvector of $\bfU$.

Suppose that  $\hat{\bfe}_1$ is perpendicular to an eigenvector
$|\bfv| = 1$ of $\bfU$ and $\hat{\bfe}_2  = \bfv \times \hat{\bfe}_1$.
Then $\hat{\bfU} :=  (-\bfI + 2 \hat{\bfe}_1 \otimes \hat{\bfe}_1)\bfU
  (-\bfI + 2 \hat{\bfe}_1 \otimes \hat{\bfe}_1) =
  (-\bfI + 2 \hat{\bfe}_2 \otimes \hat{\bfe}_2)\bfU
  (-\bfI + 2 \hat{\bfe}_2 \otimes \hat{\bfe}_2) \ne \bfU$, so that there are
by (\ref{ehati}) and (\ref{typeI_twin}) apparently four solutions
of (\ref{twin_rel2}): $\bfa_I^1 \otimes \bfn_I^1$,
$\bfa_{II}^1 \otimes \bfn_{II}^1$  based on $\hat{\bfe}_1$ and
$\bfa_I^2 \otimes \bfn_I^2$,
$\bfa_{II}^2 \otimes \bfn_{II}^2$  based on $\hat{\bfe}_2$.
By Prop. 4 of \citet{ball_james_87} these must reduce to
two.  This can happen in two possible ways:
\beq
\bfa_I^1 \parallel \bfa_{II}^2, \   \bfn_I^1 \parallel \bfn_{II}^2, \
\bfa_{II}^1 \parallel \bfa_{I}^2, \   \bfn_{II}^1 \parallel \bfn_{I}^2 \quad {\rm or}
\quad \bfa_I^1 \parallel \bfa_{II}^1,\    \bfn_I^1 \parallel \bfn_{II}^1, \
\bfa_I^2 \parallel \bfa_{II}^2,\    \bfn_I^2 \parallel \bfn_{II}^2.  \label{pars}
\eeq
By direct calculation the latter cannot happen, as it contradicts
$\hat{\bfU} \ne {\bfU}$.  The former leads to the simplification of the
formulas (\ref{typeI_twin}) given by (\ref{comp}).
\end{proof}

According to results in the Appendix, there are at most two
nonparallel unit vectors $\hat{\bfe}$ satisfying (\ref{ehati}), if
$\hat{\bfU} \ne \bfU$.  The statement to the left of the ``or'' in
(\ref{pars}) may be interpreted by saying that Compound domains are ``both
Type I and Type II'', although our precise definitions above make
Types I, II and Compound mutually exclusive.

%{\noindent {\bf The below is saved from Xian.  We may need the rotations.}
%
%\textcolor{red}{At this point, we redefine the twin(domain) relation as
%\begin{definition}\label{defn:twin_rel}
%Let $\mathbf U \in \mathbb R^{3 \times 3}_\text{+sym}$ and $\ehat \in
%\mathcal{S}^2$ be a 2-fold rotation axis such that the matrix
%$\mathbf Q = -\mathbf I + 2 \ehat \otimes \ehat$ represents a 180$^\circ$
%rotation. We call $\mathbf U$ and $\ehat$ satisfying the {\it twin or domain
%relation} with the vectors $\mathbf a \in \mathbb R^3 \backslash\{0\}$ and
%$\mathbf n\in\mathcal S^2$, if
%\begin{equation}
%\bar{\mathbf R}\mathbf U\mathbf Q = \mathbf U + \mathbf a\otimes\mathbf n
%\end{equation}
%for some rotation $\bar{\mathbf R}\in SO(3)$.
%\end{definition}}
%\begin{remark} (Saved from Xian)
% For
%Type I twin, the rotation matrix $\bar{\mathbf{R}} = - \mathbf{I} +
%\dfrac{2}{|\mathbf{U}^{-1} \ehat|^2} \mathbf{U}^{-1} \ehat
%\otimes \mathbf{U}^{-1} \ehat$. For Type II twin, the rotation matrix
%$\bar{\mathbf{R}} = - \mathbf{I} + \dfrac{2}{|\mathbf{U} \ehat|^2}
%\mathbf{U} \ehat \otimes \mathbf{U} \ehat$.
%\end{remark}
 
%$ \{\mathbf{U}\}_f\ /\ \{\mathbf{U} +
%\mathbf{a} \otimes \mathbf{n} \}_{(1-f)}$ is used to express a
%compatible laminate microstructure of martensite, in which $f$ is the
%twinning fraction of one variant versus the other.

\subsection{Crystallographic theory of martensite}
\label{crythe}

The crystallographic theory of martensite
 concerns conditions for which a twinned laminate
and the austenite phase are interpolated by a transition layer so that
the energy in the layer tends to zero as the twins are
made finer and finer.  The construction yields a sequence of
deformations $\bfy^{(k)},\ k = 1, 2, \dots$, where $k$ can be taken
as the inverse width of the transition layer,  such that
\beq
\int_\Omega \vphi(\nabla \mathbf{y}^{(k)}(\bfx), \theta_c)\, d \mathbf{x}
\to 0 \quad {\rm as}\ k \to \infty.   \label{minseq}
\eeq
Under the hypothesis of \citet[Prof. 2]{ball_james_87}, a suitable sequence $\mathbf y^{(k)}$ satisfying \eqref{minseq} converges strongly in a suitable function space to a deformation $\mathbf y$, as $k \to \infty$, such that
\beq
\nabla \mathbf y = f(\mathbf U + \mathbf a \otimes \mathbf n) + (1 - f) \mathbf U, \quad a.e.\label{weakconv}
\eeq
in the vicinity of the austenite/martensite interface and on the side of martensite. 

The equations of the crystallographic theory are built on a solution
of (\ref{twin_rel2}).  Assuming (\ref{twin_rel2}) holds,
the equations of the crystallographic theory of martensite are
\beq
\textbf{R} [f (\textbf{U} + \textbf{a} \otimes \textbf{n}) +
(1 - f) \textbf{U}] - \textbf{I} = \textbf{b} \otimes \textbf{m},
\label{ctm}
\eeq
which are to be solved for the volume fraction $0 \le f \le 1$
of the Type I/II or Compound domains,
a possible rigid rotation  $\bfR \in\ $SO(3)
of the whole martensite
laminate, and vectors
$\bfb, \bfm \in \mathbb R^{3}$.

\section{Cofactor conditions}
\label{cofcondtn}
 
The cofactor
conditions are necessary and sufficient  that (\ref{ctm})
has a solution $(f, \bfR, \bfb \otimes \bfm)$ for every $0 \le f \le 1$.
 
\begin{theorem}
\label{thm:cofactor}
Let $\mathbf U \in\mathbb R^{3\times3}_{\rm +sym}$ and define
$\hat{\bfU} = (-\bfI + 2 \hat{\bfe} \otimes \hat{\bfe})\bfU
  (-\bfI + 2 \hat{\bfe} \otimes \hat{\bfe})$ for some $|\hat{\bfe}| = 1$,
so that there exist $\hat{\bfR} \in {\rm SO(3)}$ and
$\bfa, \bfn \in \rz^3$ such that
\beq
  \hat{\bfR} \hat{\bfU} = \bfU +  \bfa \otimes \bfn.  \label{hy}
\eeq
Assume $\bfa \ne 0, \bfn \ne 0$.
The equation (\ref{ctm}) of the crystallographic theory
has a solution  $\bfR \in {\rm SO(3)}$,
$\bfb, \bfm \in \mathbb R^{3}$ for each
$f\in[0,1]$
if and only if the following {\bf cofactor conditions} are satisfied:
\begin{align*}
&\lambda_2 = 1, \text{ where $\lambda_2$ is the middle eigenvalue of
$\mathbf U$,}\tag{CC1}\label{cc1}\\
&\mathbf a\cdot\mathbf U\cof(\mathbf U^2-\mathbf I)\mathbf n = 0, \tag{CC2}\label{cc2}\\
&\textup{tr} \mathbf U^2 - \textup{det}\mathbf U^2 -
\dfrac{|\mathbf a|^2 |\mathbf n|^2}{4} - 2 \geqslant 0. \tag{CC3}\label{cc3}
\end{align*}
\end{theorem}
\begin{proof}  The proof follows Section 5 of \citet{ball_james_87}.
As is well known, e.g.,  \citet[Prop. 4]{ball_james_87}, given
$\bfU \in \mathbb R^{3\times3}_{\rm +sym}$, there is a solution
$\bfR \in {\rm SO(3)}$, $\bfc, \bfd \in \mathbb R^3$ of
$\bfR \bfU - \bfI = \bfc \otimes \bfd$ if and only if the
middle eigenvalue of $\bfU$ is 1. Since $\bfU$ has middle eigenvalue equal to 1
if and only if $\bfU^2$ has middle eigenvalue equal to 1, the satisfaction
of (\ref{ctm}) for every $0 \le f \le 1$ is equivalent to the condition
that the middle eigenvalue of the positive-definite
symmetric matrix $(\bfU + f \bfn \otimes \bfa)(\bfU +
f \bfa \otimes \bfn)$
is 1 for every $0 \le f \le 1$.  An eigenvalue of
$(\bfU + f \bfn \otimes \bfa)(\bfU + f \bfa \otimes \bfn)$ is 1
for every $0 \le f \le 1$ if and only if $g(f)$ vanishes identically
on $[0,1]$, where
\beq
g(f) = \det [(\bfU + f \bfn \otimes \bfa)(\bfU +
f \bfa \otimes \bfn ) - \bfI].
\eeq
Taking the determinant of (\ref{hy}), we see that
$\bfn \cdot \bfU^{-1}\bfa = 0$.  Hence,
$\det(\bfU + f \bfa \otimes \bfn)  = \det \bfU \ne 0$ and
\beqs
g(f) &=& (\det \bfU) \det[\bfU + f \bfa \otimes \bfn -
(\bfU + f \bfn \otimes \bfa)^{-1}] \nonumber \\
&=& (\det \bfU) \det[\bfU - \bfU^{-1} + f (\bfa \otimes \bfn +
 \bfU^{-1}\bfn \otimes \bfU^{-1}\bfa)].
\eeqs
Since the matrix multiplying $f$ is singular, then $g(f)$ is at most quadratic
in $f$.  In addition, by the hypothesis (\ref{hy}), it follows that
\beq
(\bfU + \bfn \otimes \bfa)(\bfU + \bfa \otimes \bfn)
 =  \hat{\bfU}^2 =
 (-\bfI + 2 \hat{\bfe} \otimes \hat{\bfe})
 \bfU^2 (-\bfI + 2 \hat{\bfe} \otimes \hat{\bfe}). \label{hy1}
 \eeq
Hence, putting $\bfQ = -\bfI + 2 \hat{\bfe} \otimes \hat{\bfe}$, we have that
\beq
   g(1) = \det (\bfQ \bfU^2 \bfQ^T - \bfI) = \det (\bfU^2 - \bfI) = g(0).
\eeq
A quadratic $g$ satisfying $g(0) =g(1)$ is expressible in the form
$g(f) = a(f(f-1)) + b$.  Hence, $g$ vanishes identically on $[0,1]$
if and only if $a = b = 0$.  In particular, $b = 0$ is (\ref{cc1})
and $0 = a = -g'(0)$ is (\ref{cc2}).  We have therefore shown that
(\ref{cc1}),  (\ref{cc2}) are necessary and sufficient that an
eigenvalue of $(\bfU + f \bfn \otimes \bfa)(\bfU + f \bfa \otimes \bfn)$
is 1 for every $0 \le f \le 1$.  Let the eigenvalues of
$(\bfU + f \bfn \otimes \bfa)(\bfU + f \bfa \otimes \bfn)$ be
$1, \lambda_1(f)^2, \lambda_3(f)^2$ with no particular ordering assumed.
Taking the trace of (\ref{hy1}) we have the identity
$2 \bfn \cdot \bfU \bfa  +  |\bfa|^2 |\bfn|^2 = 0$.  Using this identity
and the relations
\beqs
 1 + \lambda_1 (f)^2 + \lambda_3(f)^2 &=& {\rm tr}((\bfU +
 f \bfn \otimes \bfa)(\bfU + f \bfa \otimes \bfn))  \nonumber \\
 &=& {\rm tr}(\bfU^2) + 2 f \bfn \cdot \bfU\bfa + f^2 |\bfa|^2 |\bfn|^2,
\eeqs
and $\lambda_1(f)^2 \lambda_3(f)^2 = \det \bfU^2$, we get
\beq
 (1 - \lambda_1(f)^2)(\lambda_3(f)^2 - 1) = {\rm tr}(\bfU^2) - \det \bfU^2
  + (f^2 - f) |\bfa|^2 |\bfn|^2 - 2.  \label{eigencalc}
\eeq
Assuming (\ref{cc1}) and  (\ref{cc2}) are satisfied, (\ref{cc3}) holds
as a necessary condition that 1 is the middle eigenvalue at $f = 1/2$.
Since $f^2 - f \ge -1/4$ it is then seen that (\ref{cc1}), (\ref{cc2})
and (\ref{cc3}) are sufficient that the middle eigenvalue of
$(\bfU + f \bfn \otimes \bfa)(\bfU + f \bfa \otimes \bfn)$ is 1,
completing the proof.
\end{proof}
Noticed that $\lambda_1(f)$ and $\lambda_3(f)$ are chosen to be positive values for every $0 \leq f \leq 1$. Then it is clear that $0 < \lambda_1 = \lambda_1(0)$ and $\lambda_3 = \lambda_3(0)$ are eigenvalues of $\mathbf U$.
\begin{cor} \label{cor2}  Assume the hypotheses of Theorem \ref{thm:cofactor}
and suppose the cofactor conditions are satisfied.  Then the
other two eigenvalues $\lambda_1(f)^2 \le 1 \le \lambda_3(f)^2$
of $(\bfU + f \bfn \otimes \bfa)(\bfU +
f \bfa \otimes \bfn)$ satisfy $\lambda_1(f)^2 < 1 < \lambda_3(f)^2$
for $0 \le f \le 1$ and $f \ne 1/2$.   In particular, the
eigenvalues $\lambda_1,\ \lambda_3$
of $\bfU$ satisfy $\lambda_1 < 1< \lambda_3$.
\end{cor}

\begin{proof} Suppose  we have
some $0 \le f^* \le 1$ such that $\lambda_1 (f^*)^2 = 1$  or
$\lambda_3 (f^*)^2 = 1$.  Then, the formula  (\ref{eigencalc})
gives
\beq
0 =  (1 - \lambda_1(f^*)^2)(\lambda_3(f^*)^2 - 1) = {\rm tr}\bfU^2 - \det \bfU^2
  + ((f^*)^2 - f^*) |\bfa|^2 |\bfn|^2 - 2   \label{fstar}
\eeq
That is,
\beq
 {\rm tr}\bfU^2 - \det \bfU^2  - \dfrac{|\mathbf a|^2 |\mathbf n|^2}{4}
   - 2 = - \left((f^*)^2 - f^* + \frac{1}{4} \right) |\bfa|^2 |\bfn|^2.
   \label{eigencalc1}
\eeq
Since $(f^2 - f + \frac{1}{4}) > 0$ for $0 \le f \le 1$, $f \ne 1/2$,
then  (\ref{eigencalc1}) violates (\ref{cc3}) except
at $f^* = 1/2$, completing the proof.
\end{proof}
This result above shows incidentally that the cofactor conditions
cannot be satisfied in the classic cubic-to-tetragonal case, for in that
case the presence of a repeated eigenvalue would imply that either
$\lambda_1 = 1$ or $\lambda_3 = 1$, contradicting Corollary \ref{cor2}.

\begin{cor}  \label{cor4}
 Assume the hypotheses of Theorem \ref{thm:cofactor}
and suppose the cofactor conditions are satisfied. There
are two distinct solutions $(\bfR_f^{\kappa} \in {\rm SO(3)},\
\bfb_f^{\kappa} \otimes \bfm_f^{\kappa})$, $\kappa \in \{\pm 1\}$,
of the equation  (\ref{ctm}) of the crystallographic theory
for each $0 \le f \le 1, f \ne 1/2$.  The solutions for
$\bfb_f^{\kappa},\ \bfm_f^{\kappa}$ are
\beqs
 \bfb_f^{\kappa} &=& \frac{\rho}{\sqrt{\lambda_3(f)^2 - \lambda_1(f)^2}}
 \left( \lambda_3(f) \sqrt{1-\lambda_1(f)^2}\ \bfv_1(f)
 + \kappa \lambda_1(f) \sqrt{\lambda_3(f)^2-1}\ \bfv_3(f) \right)     \nonumber \\
 \bfm_f^{\kappa} &=& \frac{1}{\rho}\frac{\lambda_3(f) - \lambda_1(f)}
 {\sqrt{\lambda_3(f)^2 - \lambda_1(f)^2}}
 \left( - \sqrt{1-\lambda_1(f)^2}\ \bfv_1(f)  + \kappa  \sqrt{\lambda_3(f)^2-1}\
 \bfv_3(f) \right),   \label{bm}
\eeqs
$\kappa \in \{\pm 1\}$, $\rho \ne 0$ and $\bfv_1(f), \bfv_3(f)$
are orthonormal. (Note that the presence of $\rho$ does not affect
$\bfb_f^{\kappa} \otimes \bfm_f^{\kappa}$.)
\end{cor}
\begin{proof}
The existence of a solution of  (\ref{ctm}) for each $0 \le f \le 1$
follows from Theorem \ref{thm:cofactor}.  The fact that there are
two distinct solutions for $f \ne 1/2$ follows from Corollary  \ref{cor2}.
In particular, the conclusion $\lambda_1(f)^2 < 1 < \lambda_3(f)^2$ for
$f\ne 1/2$, and the explicit characterization (\ref{bm}) of the vectors
$\bfb_f^{\kappa},  \bfm_f^{\kappa}$ given by Prop. 4 of
\citet{ball_james_87}
shows that $(\bfR_f^{+1},\ \bfb_f^{+1} \otimes \bfm_f^{+1})
 \ne (\bfR_f^{-1},\ \bfb_f^{-1} \otimes \bfm_f^{-1})$.
\end{proof}
\begin{cor}
\label{cor3} Assume the hypotheses of Theorem \ref{thm:cofactor}.
In the cofactor conditions,
 (\ref{cc2}) can be replaced by the simpler form
\beq
\left(\mathbf a\cdot\hat{\mathbf v}_2\right)
\left(\mathbf n\cdot\hat{\mathbf v}_2\right) = 0,   \tag{CC2'}
\label{cc2'}
\eeq
where $\hat{\bfv}_2$ is a normalized  eigenvector of $\bfU$
corresponding to its middle eigenvalue.
That is, assuming the hypotheses of Theorem \ref{thm:cofactor},
(\ref{cc1}), (\ref{cc2}), (\ref{cc3}) $\Longleftrightarrow$
(\ref{cc1}), (\ref{cc2'}), (\ref{cc3}).
\end{cor}
\begin{proof}
Assuming the hypotheses of Theorem \ref{thm:cofactor} and
(\ref{cc1}), (\ref{cc2}), (\ref{cc3}), we write
$\bfU = \lambda_1 \hat{\bfv}_1 \otimes \hat{\bfv}_1 +
\hat{\bfv}_2 \otimes \hat{\bfv}_2 +
\lambda_3 \hat{\bfv}_3 \otimes \hat{\bfv}_3$ using ordered eigenvalues,
which, according to Corollary \ref{cor2}, satisfy
$\lambda_1 < 1 < \lambda_3$.  Then (\ref{cc3}) becomes
\beq
 (\lambda_1^2 -1)(\lambda_3^2 - 1) \left(\mathbf a\cdot\hat{\mathbf v}_2\right)
\left(\mathbf n\cdot\hat{\mathbf v}_2\right) = 0,
\eeq
implying (\ref{cc2'}).  Trivially,
(\ref{cc1}), (\ref{cc2'}), (\ref{cc3}) $\Longrightarrow$
(\ref{cc1}), (\ref{cc2}), (\ref{cc3}).
\end{proof}

\section{Microstructures possible under the cofactor conditions} \label{cof:condn}

Under the mild hypotheses of Theorem \ref{thm:cofactor}, the
satisfaction of the cofactor conditions
implies the existence of low energy transition layers in austenite/martensite
interfaces for every volume fraction
$0 \le f \le 1$, in the sense of (\ref{minseq}), i.e., in the sense of the
crystallographic theory.  In many cases the transition
layer can be eliminated altogether, resulting in zero elastic energy in these
cases.  These cases are identified here.

Let the hypotheses of Theorem \ref{thm:cofactor} be satisfied and
write the implied solutions
of the crystallographic theory as above in the form
$\bfR_f^{\kappa} \in$ SO(3), $\bfb_f^{\kappa}, \bfm_f^{\kappa} \in \rz^3$,
$\kappa \in \{\pm 1\}$, so we have
\beqs
 \hat{\bfR} \hat{\bfU} - \bfU = \bfa \otimes \bfn, \quad
 \hat{\bfU}\!\!\! &=&  \!\!\!
(-\bfI + 2 \hat{\bfe} \otimes \hat{\bfe})\bfU(-\bfI +
2 \hat{\bfe} \otimes \hat{\bfe}) ,\quad |\hat{\bfe}| = 1,  \nonumber \\
 \bfR_f^{\kappa} [ f(\bfU + \bfa \otimes \bfn) + (1-f) \bfU] &=&
 \bfI + \bfb_f^{\kappa} \otimes \bfm_f^{\kappa}, \quad 0 \le f \le 1,\ \kappa = \pm 1.
\eeqs
At $f = 0$ we have
\beq
  \bfR_0^{\kappa} \bfU = \bfI + \bfb_0^{\kappa} \otimes \bfm_0^{\kappa},
  \label{f=0case}
\eeq
which describes the implied austenite/single variant martensite interface.
 According to Corollary \ref{cor4} specialized to the
 case $f = 0 \ne 1/2$, we know there
are two distinct solutions $(\bfR_0^{\kappa} \in$ SO(3),
$\bfb_0^{\kappa} \otimes \bfm_0^{\kappa})$, $\kappa = \pm1$ of (\ref{f=0case}).
Values of $\bfb_0^{\kappa}, \bfm_0^{\kappa}$ belonging to these solutions
can be written explicitly  as
\beqs
 \bfb_0^{\kappa} &=& \frac{\rho}{\sqrt{\lambda_3^2 - \lambda_1^2}}
 \left( \lambda_3 \sqrt{1-\lambda_1^2}\ \bfv_1
 + \kappa \lambda_1 \sqrt{\lambda_3^2-1}\ \bfv_3\right)     \nonumber \\
 \bfm_0^{\kappa} &=& \frac{1}{\rho}\frac{\lambda_3 - \lambda_1}{\sqrt{\lambda_3^2 -
 \lambda_1^2}}
 \left( - \sqrt{1-\lambda_1^2}\ \bfv_1  + \kappa  \sqrt{\lambda_3^2-1}\
 \bfv_3 \right),  \quad \kappa \in \{\pm 1\}, \label{b0m0}
\eeqs
for some $\rho \ne 0$ by specialization of (\ref{bm}), where
$0 < \lambda_1 < 1 < \lambda_3$ are the ordered eigenvalues of $\bfU$
with corresponding orthonormal eigenvectors $\bfv_1, \bfv_2, \bfv_3$.

\subsection{Preliminary results for Types I and II domains}
 
Proposition \ref{lem1} says that if the cofactor conditions
are satisfied for Type I or Type II domains, then
$\hat{\bfU} = (-\bfI + 2 \hat{\bfe} \otimes \hat{\bfe})\bfU
(-\bfI + 2 \hat{\bfe} \otimes \hat{\bfe})$ holds for some
$\hat{\bfe}$ with $\bfv_2 \cdot \hat{\bfe} \ne 0$.
In fact, only one unit vector $\hat{\bfe}$ satisfies this
condition up to $\pm$.

The condition $\bfv_2 \cdot \hat{\bfe} \ne 0$ implies that the main
condition (\ref{cc2'}) (see  Corollary \ref{cor3}) of the cofactor
conditions simplifies for Types I and II domains.

\begin{proposition} \label{propI/II}
Assume $\bfU = \lambda_1 \bfv_1 \otimes \bfv_1 + \bfv_2 \otimes \bfv_2
+ \lambda _3 \bfv_3 \otimes \bfv_3 $, $0 < \lambda_1 < 1 < \lambda_3$,   and
 $\hat{\bfU} = (-\bfI + 2 \hat{\bfe} \otimes \hat{\bfe})\bfU
  (-\bfI + 2 \hat{\bfe} \otimes \hat{\bfe}) \ne \bfU$, $|\hat{\bfe}| = 1$.
  Recall Corollary \ref{cor3}.
\begin{enumerate}
\item For Type I domains
$(\bfa_I \cdot \bfv_2)(\bfn_I \cdot \bfv_2) = 0 \Longleftrightarrow
\bfa_I \cdot \bfv_2 =0 \Longleftrightarrow|\bfU^{-1} \hat{\bfe}| = 1$. \label{Ichar}
\item For Type II domains
$(\bfa_{II} \cdot \bfv_2)(\bfn_{II} \cdot \bfv_2) = 0
\Longleftrightarrow \bfn_{II} \cdot \bfv_2 = 0 \Longleftrightarrow|\bfU\, \hat{\bfe}| = 1$.
\label{IIchar}
\end{enumerate}
\end{proposition}
\begin{proof} By  Proposition \ref{lem1} and the definitions
of Type I and II domains (which exclude the case of Compound domains), we
have $\hat{\bfe} \cdot \bfv_2 \ne 0$. The results then follow from
(\ref{typeI_twin})  and the condition
$\bfU \bfv_2  = \bfv_2$.
\end{proof}

Proposition \ref{propI/II} shows that one of the two main
cofactor conditions can be interpreted geometrically as the
condition that the vector $\hat{\bfe}$ which defines the twin
system (or, more generally, the domain system) lies on the intersection of
the strain ellipsoid, or inverse strain ellipsoid, and the unit sphere.

\subsection{Elimination of the transition layer in the austenite/martensite
interface for some Type I
domains}

The removal of the transition layer in the case of Type I domains
proceeds by proving the existence of a zero-energy triple junction.
The key is to prove that
$\bfR_1^{\kappa_*}  = \bfR_0^{\kappa}$ for suitable choices of
$\kappa, \kappa_* \in \{\pm1\}$.

\begin{theorem} \label{thmI} (Type I domains)  Assume the hypotheses of Theorem
\ref{thm:cofactor} and suppose the cofactor conditions are satisfied
using Type I domains.
There are particular choices of $\sigma, \sigma_* \in \{\pm 1\}$ such that
$\bfR_1^{\sigma_*}  = \bfR_0^{\sigma}$ and
$\bfb^{\sigma_*}_1  =  \xi\, \bfb^{\sigma}_0$ for some $\xi \ne 0$, so that
\beq
\bfR_0^{\sigma} \bfU  = \bfI +
\bfb_0^{\sigma} \otimes \bfm_0^{\sigma}, \quad
\bfR_0^{\sigma} (\bfU + \bfa_I \otimes \bfn_I) =
 \bfI + \bfb_0^{\sigma} \otimes \xi \bfm_1^{\sigma_*},   \label{cases}
\eeq
and therefore, by taking a convex combination of the equations in
(\ref{cases}), one of the two families of solutions of the
crystallographic theory can be written
\beq
 \bfR_0^{\sigma} [\bfU + f \bfa_I \otimes \bfn_I)] =
 \bfI + \bfb_0^{\sigma} \otimes \bigg(f \xi \bfm_1^{\sigma_*} +
 (1-f) \bfm_0^{\sigma} \bigg)
 \quad {\rm for\ all}\ 0 \le f \le 1.  \label{cc}
\eeq
The three deformation gradients
$\bfI,\ \bfR_0^{\sigma} \bfU,\ \bfR_0^{\sigma}\hat{\bfR} \hat{\bfU}$
can form a compatible austenite/martensite triple junction
in the sense that
\beq
\bfR_0^{\sigma} \bfU - \bfI = \bfb_0^{\sigma} \otimes \bfm_0^{\sigma},
\quad
\bfR_0^{\sigma}\hat{\bfR} \hat{\bfU} - \bfI = \bfb_0^{\sigma} \otimes \xi \bfm_1^{\sigma_*},
\quad
\bfR_0^{\sigma}\hat{\bfR} \hat{\bfU} - \bfR_0^{\sigma} \bfU = \bfR_0^{\sigma} \bfa_I \otimes \bfn_I.
\label{triple}
\eeq
There is a constant $c \ne 0$ such that $ c \bfn_I  = \xi \bfm_1^{\sigma_*} - \bfm_0^{\sigma}$,
so the three vectors $\bfm_0^{\sigma}, \bfm_1^{\sigma_*}$, and $\bfn_I$ lie in a plane.
\end{theorem}
\begin{proof}  By Proposition \ref{propI/II} we have for Type I domains
under the cofactor conditions,  $\bfa_I \cdot \bfv_2 = 0$ and
$|\bfU^{-1} \hat{\bfe}| = |\hat{\bfe}| = 1$.  The latter can be
written, alternatively,
\beq
 \hat{\bfe} \cdot (\bfU^{-2} - \bfI) \hat{\bfe} = 0
 \Longleftrightarrow \lambda_3 \sqrt{1- \lambda_1^2}\ (\bfv_1 \cdot \hat{\bfe})
 = \pm \lambda_1 \sqrt{\lambda_3^2 - 1} \ (\bfv_3 \cdot \hat{\bfe}). \label{Iconds}
\eeq
Note in passing that $\bfv_3 \cdot \hat{\bfe} \ne 0$, because, if this
were not the case, then it would follow by (\ref{Iconds}) and
Corollary \ref{cor2} that also  $\bfv_1 \cdot \hat{\bfe} = 0$,
so $\hat{\bfe} \parallel \bfv_2$.  But then it would follow that
$\hat{\bfU} = (-\bfI + 2 \hat{\bfe} \otimes \hat{\bfe}) \bfU
(-\bfI + 2 \hat{\bfe} \otimes \hat{\bfe}) = \bfU$ which is
forbidden.

By Corollary \ref{cor4}, we have two families of solutions of the
crystallographic theory that can be written
 $(\bfR_f^{\kappa} \in {\rm SO(3)},\
\bfb_f^{\kappa} \otimes \bfm_f^{\kappa})$, $\kappa \in \{\pm 1\}$,
$0 \le f \le 1$ and these are distinct if $f \ne 1/2$. Thus, at
$f = 1$,
\beq
\bfR_1^{\kappa} (\bfU + \bfa_I \otimes \bfn_I) =
\bfR_1^{\kappa} \hat{\bfR} \hat{\bfU} = \bfI +
\bfb_1^{\kappa} \otimes \bfm_1^{\kappa},  \quad \kappa \in \{\pm 1\}.
\label{f1}
\eeq
Using that $\hat{\bfU} = (-\bfI +
2 \hat{\bfe} \otimes \hat{\bfe}) \bfU (-\bfI + 2\hat{\bfe} \otimes \hat{\bfe})$
and pre- and post- multiplying (\ref{f1}) by the 180 degree rotation
$\hat{\bfQ} = (-\bfI +
2 \hat{\bfe} \otimes \hat{\bfe}) = \hat{\bfQ}^T$, we have that
\beq
\hat{\bfQ}\bfR_1^{\kappa} \hat{\bfR} \hat{\bfQ} \bfU = \bfI +
\hat{\bfQ}\bfb_1^{\kappa} \otimes \hat{\bfQ}\bfm_1^{\kappa},
\quad \kappa \in \{\pm 1\} \label{f=1case}
\eeq
Comparison of (\ref{f=1case}) with (\ref{f=0case}) shows that there is
a map $\hat{\sigma}: \{\pm 1\} \to \{\pm 1\}$ and $\delta \ne 0$ such that
$\hat{\bfQ} \bfb_1^{\hat{\sigma}(\kappa)} = \delta \bfb_0^{\kappa},\
\hat{\bfQ}  \bfm_1^{\hat{\sigma}(\kappa)}  =
(1/\delta) \bfm_0^{\kappa}$, i.e.,
\beq
\bfb_1^{\hat{\sigma}(\kappa)} =
\delta (-\bfI +
2 \hat{\bfe} \otimes \hat{\bfe})\bfb_0^{\kappa}, \quad
\bfm_1^{\hat{\sigma}(\kappa)}  =
\frac{1}{\delta} (-\bfI +
2 \hat{\bfe} \otimes \hat{\bfe}) \bfm_0^{\kappa}. \label{b1m1}
\eeq
We note from (\ref{f=0case}), (\ref{b0m0}) and (\ref{Iconds})  that
\beqs
   \bfb_0^{\kappa} \cdot \hat{\bfe}  &=& \frac{\rho}{\sqrt{\lambda_3^2 - \lambda_1^2}}
 \left( \lambda_3 \sqrt{1-\lambda_1^2} (\bfv_1 \cdot \hat{\bfe})
 + \kappa \lambda_1 \sqrt{\lambda_3^2-1} (\bfv_3 \cdot \hat{\bfe})\right)  \nonumber \\
&=& \frac{\rho  \lambda_1 \sqrt{\lambda_3^2-1} (\bfv_3 \cdot \hat{\bfe})}
{\sqrt{\lambda_3^2 - \lambda_1^2}}  (\pm 1 + \kappa).
\eeqs
Hence there is a particular choice $\kappa = \sigma \in \{ \pm 1 \}$
such that  $\bfb_0^{\sigma} \cdot \hat{\bfe} = 0$. Let
$\sigma_* = \hat{\sigma}(\sigma)$.  For these choices we have from
(\ref{b1m1}) that
\beq
\bfb_1^{\sigma_*} =
-\delta \bfb_0^{\sigma},  \label{b1m1-}
\eeq
so, in particular, $\bfb_1^{\sigma_*} \cdot \hat{\bfe} =
\bfb_1^{\sigma_*} \cdot \bfv_2 = 0$.

Take the determinant of (\ref{f1}) to observe that
 $1+ \bfb_1^{\sigma_*} \cdot \bfm_1^{\sigma_*} =
\det \bfR_1^{\sigma_*} \hat{\bfR} \hat{\bfU} = \det \bfU >0$.
Premultiply (\ref{f1}) by $(\bfR_1^{\sigma_*})^T$, take the transpose of the resulting
equation, operate the result on $\bfv_2$, and use that
$\bfU\bfv_2 = \bfv_2$ and $\bfa_I \cdot \bfv_2 = 0$
(Proposition \ref{propI/II}) to get
\beq
 \bfR_1^{\sigma_*}\bfv_2 = \bfv_2 -
( \bfb_1^{\sigma_*}\cdot \bfR_1^{\sigma_*}\bfv_2 ) \bfm_1^{\sigma_*}.
\label{interm}
\eeq
Dot (\ref{interm}) with $\bfb_1^{\sigma_*}$ and use that
$1 + \bfb_1^{\sigma_*} \cdot \bfm_1^{\sigma_*} > 0$:
\beq
\bfb_1^{\sigma_*} \cdot  \bfR_1^{\sigma_*}\bfv_2
= \frac{1}{( 1 +
 \bfb_1^{\sigma_*} \cdot \bfm_1^{\sigma_*})}\bfb_1^{\sigma_*} \cdot \bfv_2 = 0.
\label{interm1}
\eeq
(The latter follows from (\ref{b1m1-}).)  Equations (\ref{interm}) and
(\ref{interm1}) show that $ \bfR_1^{\sigma_*}\bfv_2 = \bfv_2$.  Using this
conclusion and $\bfn_I = \hat{\bfe}$, evaluate (\ref{f1}) at
$\kappa = \sigma_*$ and operate the
result on $\bfv_2$ to get
\beq
(\hat{\bfe} \cdot \bfv_2)\bfR_1^{\sigma_*} \bfa_I  =
(\bfm_1^{\sigma_*} \cdot \bfv_2) \bfb_1^{\sigma_*}  =
-\delta (\bfm_1^{\sigma_*} \cdot \bfv_2) \bfb_0^{\sigma}.  \label{f1red}
\eeq
Proposition \ref{lem1} shows that $\hat{\bfe} \cdot \bfv_2 \ne 0$, so
both sides of (\ref{f1red}) are nonvanishing.  Thus we can condense the
constants by writing
$\bfR_1^{\sigma_*} \bfa_I = c  \bfb_0^{\sigma}$  for some $c \ne 0$.
Substitution of the latter back into (\ref{f1}) ($\kappa = \sigma_*$)
and use of (\ref{b1m1-}) gives
\beq
\bfR_1^{\sigma_*} \bfU = \bfI
+ \bfb_0^{\sigma} \otimes (-\delta \bfm_1^{\sigma_*} - c \bfn_I).
\label{fI2}
\eeq
Comparison of (\ref{fI2}) and (\ref{b0m0})
(note: $\bfb_0^{+1} \nparallel \bfb_0^{-1}$ under our hypotheses)
we get that
\beq
 \bfR_1^{\sigma_*} = \bfR_0^{\sigma} \quad {\rm and} \quad
 \delta \bfm_1^{\sigma_*} + c \bfn_I  = - \bfm_0^{\sigma}.  \label{lastI}
\eeq
We have proved Theorem \ref{thmI} up to (\ref{cases}), and (\ref{cc})
is $(1-f)$(\ref{cases})$_1 + f$(\ref{cases})$_2$.  The three rank-one
connections summarized in (\ref{triple}) are from (\ref{cases})
and the basic rank-one relation (\ref{twin_rel2})-(\ref{typeI_twin}).
The planarity of the three vectors follows from (\ref{lastI}).
\end{proof}

\begin{figure}[htbp]
\begin{center}
\subfigure{\includegraphics[width = 0.35 \textwidth]{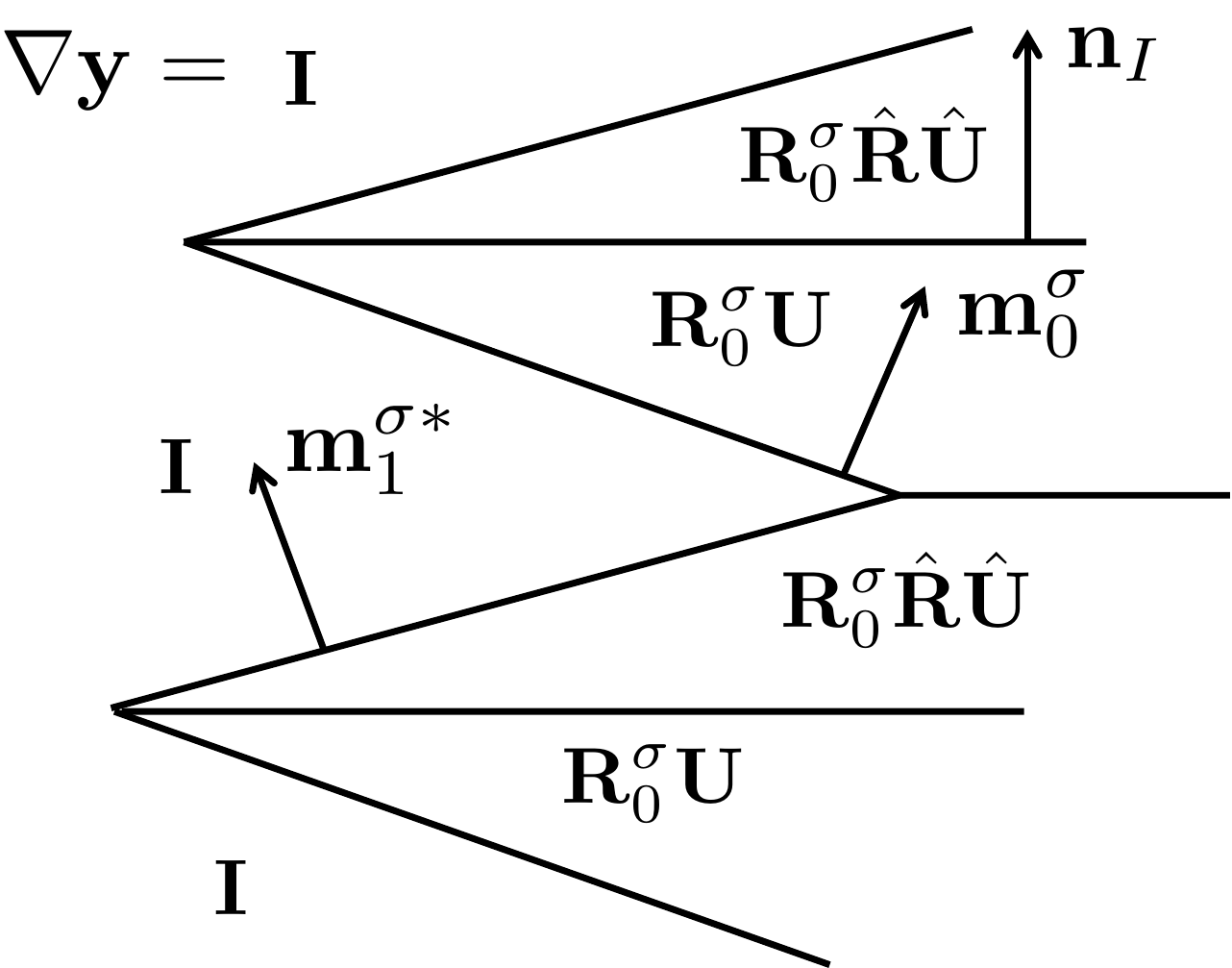}}\hspace{20pt}
\subfigure{\includegraphics[width = 0.4 \textwidth]{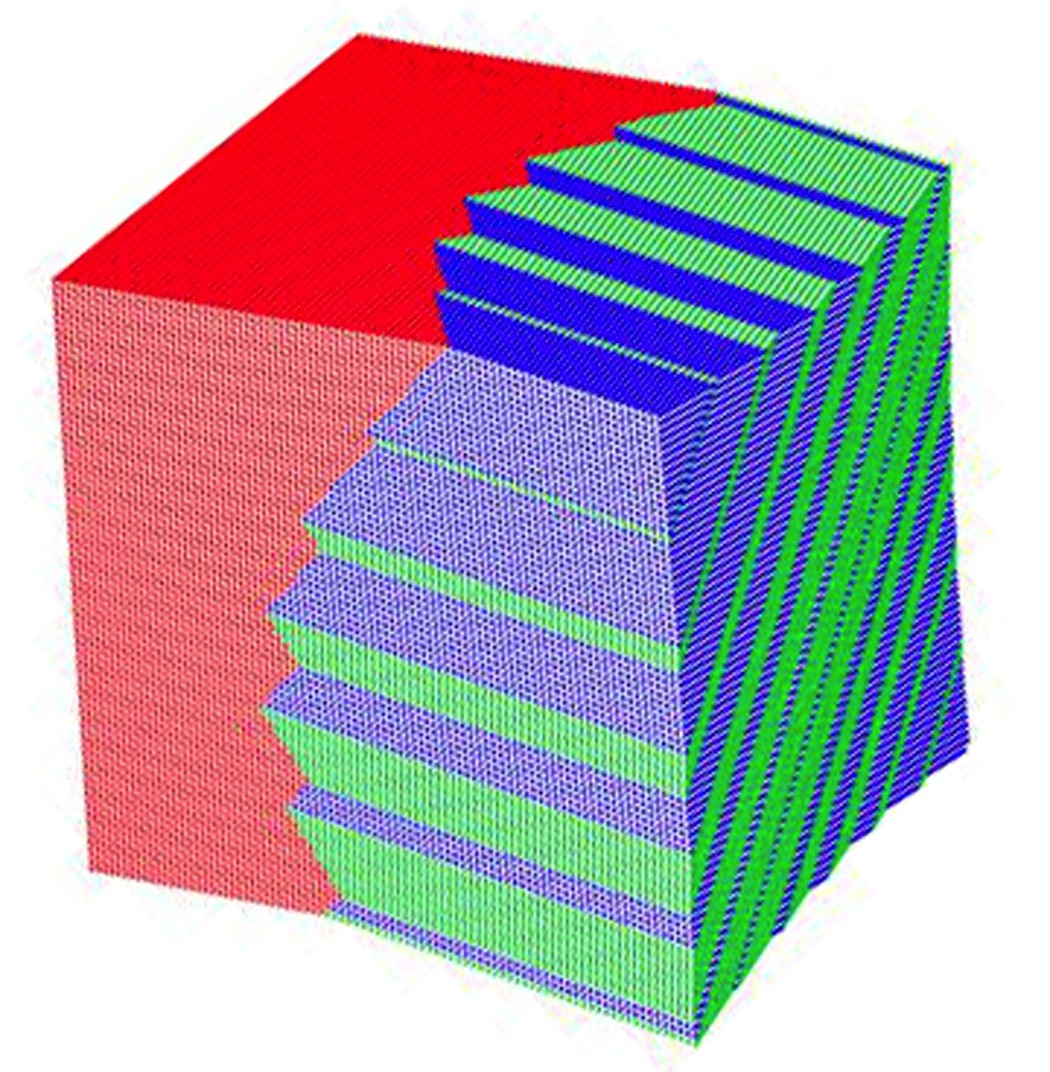}}\hspace{20pt}
\caption{Left diagram is a schematic of three triple conjunctions using the deformation gradients in \eqref{triple}. A macroscopically curved austenite/martensite interface with zero elastic energy is plotted on the right for a material satisfying the cofactor conditions (Type I domain).}
\label{curve_inf}
\end{center}
\end{figure}

Several remarks are worth noting.  First, the final statement about the
planarity of the three vectors is important for actually making the
indicated triple junction.   Second,  the solutions of the crystallographic
theory  given
by (\ref{cc}) do not necessarily correspond to the choice
$\kappa = \sigma$ for all $0 \le f\le 1$ in Corollary
\ref{cor4}.   In fact, the numerical evidence
supports the idea that the solution found in Theorem \ref{thmI}
agrees with different choices of $\kappa$ in Corollary \ref{cor4} for
different values of $f$, although this can be fixed by choosing eigenvectors
$\bfv_1(f),\bfv_2(f), \bfv_3(f)$ that change continuously with $f$ (This, of course, is not done by most
numerical packages).
Third, in the arguments of Theorem \ref{thmI}
 we have nowhere used
the inequality (\ref{cc3}) of the cofactor conditions.  Hence, the
particular family solutions of the crystallographic theory found here
does not rely on explicitly assuming this inequality.  In fact, the inequality
(\ref{cc3}) can be proved as a necessary condition
by use of (\ref{bm}) and (\ref{f1}).

\begin{figure}[ht]
\begin{center}
\subfigure{\includegraphics[width = 5 in]{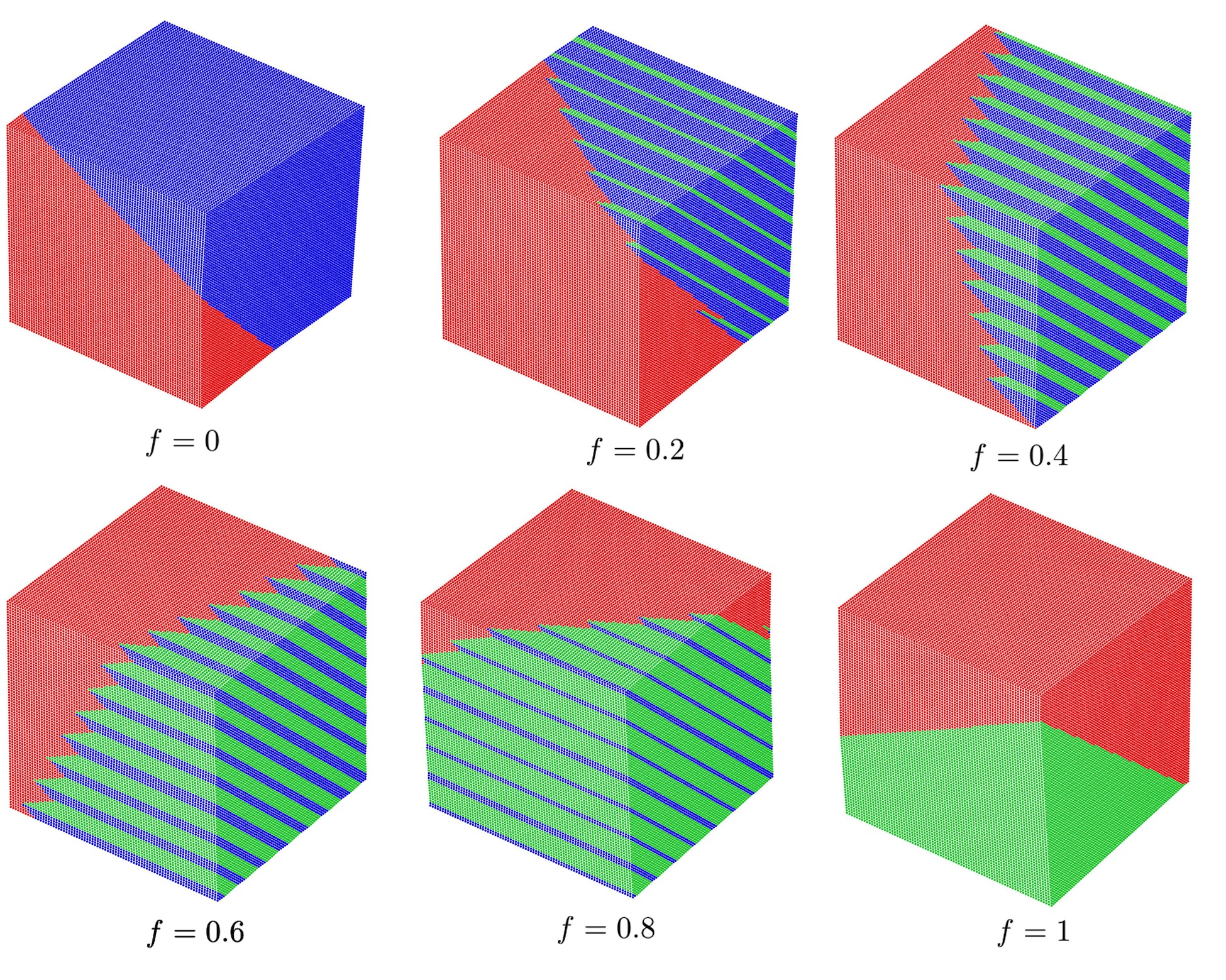}}\hspace{20pt}
%\subfigure{\includegraphics[width = 0.8 \textwidth]{type2twin_line.jpg}}\hspace{20pt}
\caption{Zero elastic energy austenite/martensite interfaces for a
material satisfying the
cofactor conditions (Type I domain) at  various $f$ from 0 to 1.}
\label{typeI_interface}
\end{center}
\end{figure}
The compatibility conditions given in \eqref{triple} imply the existence of several
interesting microstructures using the triple junction as a building block.  Figure \ref{curve_inf} (left) gives a schematic of  three
triple junctions.  Note that by \eqref{triple} all the jump conditions across all interfaces are satisfied.  Satisfaction of all such jump conditions implies the existence of a continuous deformation with these gradients.  Examples of
deformations constructed in this way (using the method of visualization described in the introduction) are shown in Figures \ref{curve_inf} (right), \ref{typeI_interface}, \ref{nucleation_AinM} and \ref{nucleation_MinA}. Figure \ref{typeI_interface} shows the configurations of austenite/martensite interfaces having
zero elastic energy for $f$ varying from
0 to 1. 

\subsection{Elimination of the transition layer in the austenite/martensite
interface for some Type II
domains}

The reason for the elimination of the transition layer
in the case of Type II domains is different --
it arises from the parallelism of a single variant martensite/austenite
interface and a domain wall -- but the mathematical argument is dual
to the argument for Type I domains.

\begin{figure}[ht]
\begin{center}
\subfigure{\includegraphics[width = 0.7 \textwidth]{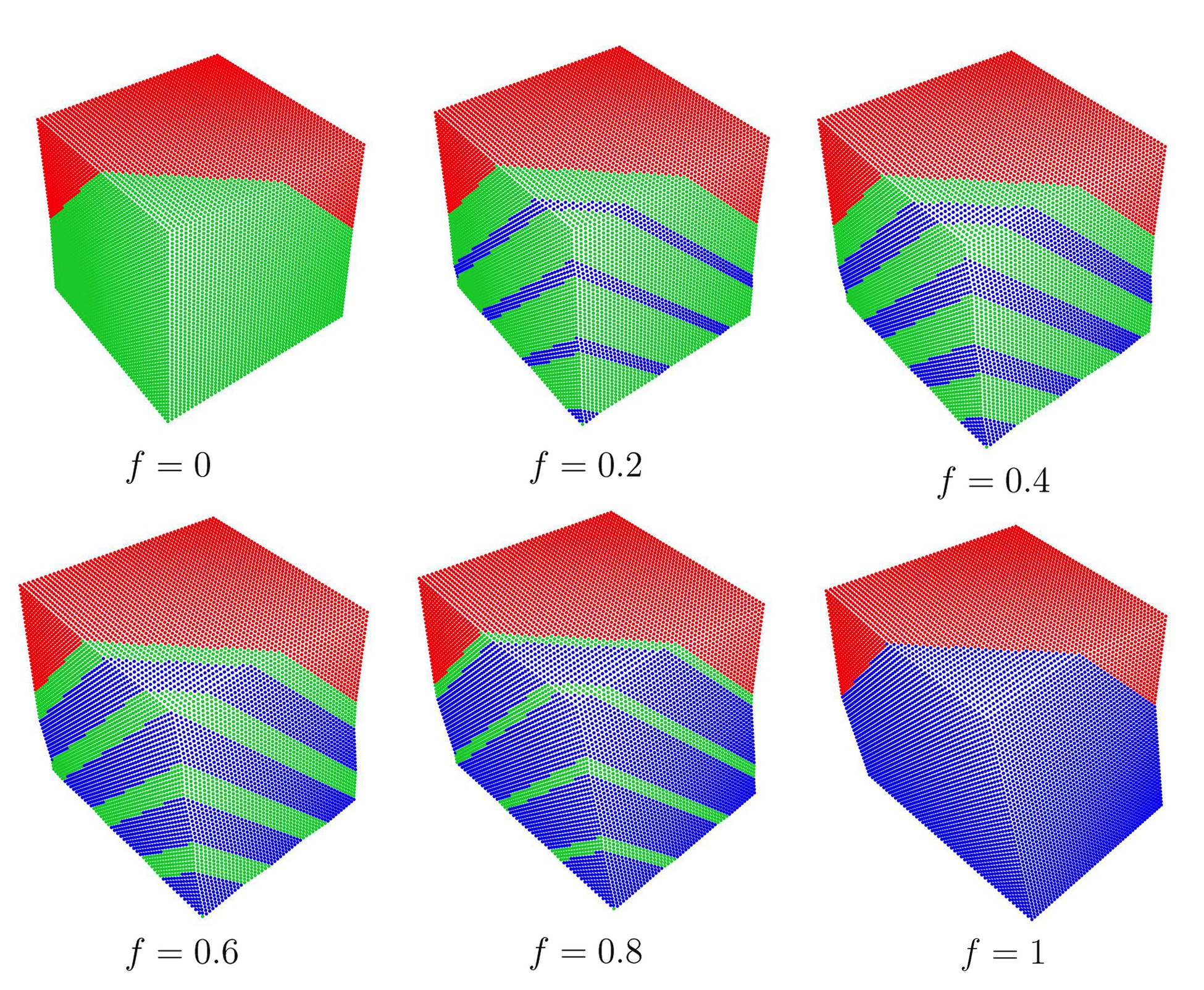}}\hspace{20pt}
%\subfigure{\includegraphics[width = 0.8 \textwidth]{type2twin_line.jpg}}\hspace{20pt}
\caption{Zero elastic energy austenite/martensite interfaces for a
material satisfying the
cofactor conditions (Type II domain) at  various $f$ from 0 to 1.}
\label{type2_interface}
\end{center}
\end{figure}

\begin{theorem} \label{thmII} (Type II domains)  Assume the hypotheses of Theorem
\ref{thm:cofactor} and suppose the cofactor conditions are satisfied
using Type II domains.
There are particular choices of $\sigma, \sigma_* \in \{\pm 1\}$ such that
$\bfR_1^{\sigma_*}  = \bfR_0^{\sigma}$ and
$\bfm^{\sigma_*}_1  =  \xi\, \bfm^{\sigma}_0$ for some $\xi \ne 0$, so that
\beq
\bfR_0^{\sigma} \bfU  = \bfI +
\bfb_0^{\sigma} \otimes \bfm_0^{\sigma}, \quad
\bfR_0^{\sigma} (\bfU + \bfa_{II} \otimes \bfn_{II}) =
 \bfI + \xi \bfb_1^{\sigma_*} \otimes \bfm_0^{\sigma},   \label{casesII}
\eeq
and therefore, by taking a convex combination of the equations in
(\ref{casesII}), one of the two families of solutions of the
crystallographic theory can be written
\beq
 \bfR_0^{\sigma} [\bfU + f \bfa_{II} \otimes \bfn_{II})] =
 \bfI +  \bigg(f \xi \bfb_1^{\sigma_*} +
 (1-f) \bfb_0^{\sigma} \bigg)  \otimes \bfm_0^{\sigma}
 \quad {\rm for\ all}\ 0 \le f \le 1.  \label{ccII}
\eeq
The normal $\bfm_0^{\sigma}$ to the austenite/martensite interface is
independent of the volume fraction $f$ and is parallel to the
domain wall normal: $ \bfn_{II} = c\,  \bfm_0^{\sigma}$ for some $c \ne 0$.
\end{theorem}

\begin{proof}
By Proposition \ref{propI/II} we have for Type II domains
under the cofactor conditions,  $\bfn_{II} \cdot \bfv_2 = 0$ and
$|\bfU \hat{\bfe}|^2 = |\hat{\bfe}|^2 = 1$.  The latter can be
written
\beq
 \hat{\bfe} \cdot (\bfU^{2} - \bfI) \hat{\bfe} = 0
 \Longleftrightarrow  \sqrt{1- \lambda_1^2}\ (\bfv_1 \cdot \hat{\bfe})
 = \pm  \sqrt{\lambda_3^2 - 1} \ (\bfv_3 \cdot \hat{\bfe}), \label{IIconds}
\eeq
and, as above,  $\bfv_3 \cdot \hat{\bfe} \ne 0$.

Recycling the notation of the Type I case, we have
two families of solutions of the
crystallographic theory that can be written
 $(\bfR_f^{\kappa} \in {\rm SO(3)},\
\bfb_f^{\kappa} \otimes \bfm_f^{\kappa})$, $\kappa \in \{\pm 1\}$,
$0 \le f \le 1$ and these are distinct if $f \ne 1/2$. Thus, at
$f = 1$,
\beq
\bfR_1^{\kappa} (\bfU + \bfa_{II} \otimes \bfn_{II}) =
\bfR_1^{\kappa} \hat{\bfR} \hat{\bfU} = \bfI +
\bfb_1^{\kappa} \otimes \bfm_1^{\kappa},  \quad \kappa \in \{\pm 1\}.
\label{f1II}
\eeq
Using that $\hat{\bfU} = (-\bfI +
2 \hat{\bfe} \otimes \hat{\bfe}) \bfU (-\bfI + 2\hat{\bfe} \otimes \hat{\bfe})$
and pre- and post- multiplying (\ref{f1}) by the 180 degree rotation
$\hat{\bfQ} = (-\bfI +
2 \hat{\bfe} \otimes \hat{\bfe}) = \hat{\bfQ}^T$, we have that
\beq
\hat{\bfQ}\bfR_1^{\kappa} \hat{\bfR} \hat{\bfQ} \bfU = \bfI +
\hat{\bfQ}\bfb_1^{\kappa} \otimes \hat{\bfQ}\bfm_1^{\kappa},
\quad \kappa \in \{\pm 1\} \label{f=1caseII}
\eeq
Comparison of (\ref{f=1caseII}) with (\ref{f=0case}) shows that there is
a map $\hat{\sigma}: \{\pm 1\} \to \{\pm 1\}$ and $\delta \ne 0$ such that
$\hat{\bfQ} \bfb_1^{\hat{\sigma}(\kappa)} = \delta \bfb_0^{\kappa},\
\hat{\bfQ}  \bfm_1^{\hat{\sigma}(\kappa)}  =
(1/\delta) \bfm_0^{\kappa}$, i.e.,
\beq
\bfb_1^{\hat{\sigma}(\kappa)} =
\delta (-\bfI +
2 \hat{\bfe} \otimes \hat{\bfe})\bfb_0^{\kappa}, \quad
\bfm_1^{\hat{\sigma}(\kappa)}  =
\frac{1}{\delta} (-\bfI +
2 \hat{\bfe} \otimes \hat{\bfe}) \bfm_0^{\kappa}. \label{b1m1II}
\eeq
We note from (\ref{f=0case}), (\ref{b0m0}) and (\ref{IIconds})  that
\beqs
\bfm_0^{\kappa} \cdot \hat{\bfe} &=& \frac{1}{\rho}\frac{\lambda_3 -
\lambda_1}{\sqrt{\lambda_3^2 - \lambda_1^2}}
 \left( - \sqrt{1-\lambda_1^2}\ (\bfv_1 \cdot \hat{\bfe})  + \kappa  \sqrt{\lambda_3^2-1}\
 (\bfv_3 \cdot \hat{\bfe}) \right),   \nonumber \\
  &=&
\frac{1}{\rho}\frac{\sqrt{\lambda_3^2-1}(\lambda_3 -
\lambda_1)(\bfv_3 \cdot \hat{\bfe})}{\sqrt{\lambda_3^2 - \lambda_1^2}}
(\mp 1 + \kappa), \quad \kappa \in \{+1, -1\} .
\eeqs
Hence there is a particular choice $\kappa = \sigma \in \{ \pm 1 \}$
such that  $\bfm_0^{\sigma} \cdot \hat{\bfe} = 0$. Let
$\sigma_* = \hat{\sigma}(\sigma)$.  For these choices we have from
(\ref{b1m1II}) that
\beq
\bfm_1^{\sigma_*} =
-\frac{1}{\delta} \bfm_0^{\sigma},  \label{b1m1-II}
\eeq
so, in particular, $\bfm_1^{\sigma_*} \cdot \hat{\bfe} =
\bfm_1^{\sigma_*} \cdot \bfv_2 = 0$.

%Take the determinant of (\ref{f1II}) to observe that
% $1+ \bfb_1^{\sigma_*} \cdot \bfm_1^{\sigma_*} =
%\det \bfR_1^{\sigma_*} \hat{\bfR} \hat{\bfU} = \det \bfU >0$.
%Premultiply (\ref{f1}) by $(\bfR_1^{\sigma_*})^T$, take the transpose of the resulting
%equation, operate the result on $\bfv_2$, and use that
%$\bfU\bfv_2 = \bfv_2$ and $\bfa_I \cdot \bfv_2 = 0$
%(Proposition \ref{propI/II}) to get
Following the dual of the Type I case, evaluate (\ref{f1II}) at
$\kappa = \sigma_*$ and operate on $\bfv_2$ to get
\beq
 \bfR_1^{\sigma_*}\bfv_2 = \bfv_2 +
( \bfm_1^{\sigma_*}\cdot \bfv_2 ) \bfb_1^{\sigma_*} = \bfv_2.
\label{intermII}
\eeq
 Using the formula (\ref{typeI_twin}) for $\bfa_{II}$,
 evaluate (\ref{f1II}) at
$\kappa = \sigma_*$ and operate its transpose on $\bfv_2$ to get
\beq
(\bfa_{II} \cdot \bfv_2)\bfn_{II} =
(\bfb_1^{\sigma_*} \cdot \bfv_2)\bfm_1^{\sigma_*}.  \label{f1redII}
\eeq
Lemma \ref{lem1} shows that $\bfa_{II} \cdot \bfv_2 = \hat{\bfe} \cdot \bfv_2 \ne 0$, so
both sides of (\ref{f1redII}) are nonvanishing.  Thus we can condense the
constants by writing
$ \bfn_{II} = c  \bfm_0^{\sigma}$  for some $c \ne 0$.
Substitution of the latter back into (\ref{f1II}) ($\kappa = \sigma_*$)
and use of (\ref{b1m1-II}) gives
\beq
\bfR_1^{\sigma_*} \bfU = \bfI
+ (- c \bfR_1^{\sigma_*} \bfa_{II} - \frac{1}{\delta}\bfb_1^{\sigma_*})
\otimes \bfm_0^{\sigma}.
\label{f2}
\eeq
Comparison of (\ref{f2}) and (\ref{b0m0})
(note: $\bfb_0^{+1} \nparallel \bfb_0^{-1}$ under our hypotheses)
we get that
\beq
 \bfR_1^{\sigma_*} = \bfR_0^{\sigma} \quad {\rm and} \quad
\bfR_1^{\sigma_*} \bfa_{II} + \frac{1}{\delta}\bfb_1^{\sigma_*}
= - \bfb_0^{\sigma}.  \label{lastII}
\eeq
We have proved Theorem \ref{thmII} up to (\ref{casesII}), and (\ref{ccII})
is $(1-f)$(\ref{casesII})$_1 + f$(\ref{casesII})$_2$.  The parallelism of
$\bfn_{II}$ and $\bfm_0^{\sigma}$ is (\ref{f1redII}).
\end{proof}
Some of the remarks following the proof of Theorem \ref{thmI} apply
here as well.   In a certain sense
these results show that, under the cofactor conditions,
triple junctions are dual to parallel austenite/twin interfaces.
The duality is that which maps Type I into Type II twins.

\subsection{The cofactor conditions for Compound domains}

We assume in this subsection the hypotheses of Proposition \ref{lem1} which
gives the basic characterization of Compound domains.  Specifically, we
assume that there are orthonormal vectors $\hat{\bfe}_1, \hat{\bfe}_2$
such that $\hat{\bfU} = (-\bfI + 2 \hat{\bfe}_1 \otimes \hat{\bfe}_1)\bfU
(-\bfI + 2 \hat{\bfe}_1 \otimes \hat{\bfe}_1) = (-\bfI + 2 \hat{\bfe}_2 \otimes \hat{\bfe}_2)\bfU
(-\bfI + 2 \hat{\bfe}_2 \otimes \hat{\bfe}_2) \ne \bfU$.  The  two solutions
of (\ref{twin_rel2}) for Compound domains
$\bfa_C^1 \otimes \bfn_C^1,\ \bfa_C^2 \otimes \bfn_C^2$
are then given by (\ref{comp}).

\begin{lem} \label{lemcomp}  Suppose that there are orthonormal vectors $\hat{\bfe}_1, \hat{\bfe}_2$
such that $\hat{\bfU} = (-\bfI + 2 \hat{\bfe}_1 \otimes \hat{\bfe}_1)\bfU
(-\bfI + 2 \hat{\bfe}_1 \otimes \hat{\bfe}_1) = (-\bfI + 2 \hat{\bfe}_2 \otimes \hat{\bfe}_2)\bfU
(-\bfI + 2 \hat{\bfe}_2 \otimes \hat{\bfe}_2) \ne \bfU$, and let
$\bfa_C^1 \otimes \bfn_C^1,\ \bfa_C^2 \otimes \bfn_C^2$
be given by (\ref{comp}). The cofactor conditions are satisfied for
either of these solutions if and only if $\hat{\bfe}_1 \cdot \bfv_2 = 0$,
$\hat{\bfe}_2 \cdot \bfv_2 = 0$, $\hat{\bfe}_1 $ is not parallel
to either $\bfv_1$ or $\bfv_3$, and the inequality (\ref{cc3}) holds.
\end{lem}
\begin{proof}  By Corollary \ref{cor3},
the condition (\ref{cc2})
of the cofactor conditions
for either solution
$\bfa_C^1 \otimes \bfn_C^1$ or  $\bfa_C^2 \otimes \bfn_C^2$  reduces to
\beq
(\hat{\bfe}_1 \cdot \bfv_2)(\hat{\bfe}_2 \cdot \bfv_2) = 0. \label{compccc3}
\eeq
Suppose the cofactor conditions are satisfied.
According to Proposition \ref{comp} both
$\hat{\bfe}_1$ and  $\hat{\bfe}_2$ are perpendicular to an eigenvector
of $\bfU$.  But this eigenvector cannot be $\bfv_1$ or $\bfv_3$, because
then (\ref{compccc3})  would force either $\hat{\bfe}_1$ or
$\hat{\bfe}_2$ to be parallel to
an eigenvector of $\bfU$ which contradicts $\hat{\bfU} \ne \bfU$.
Therefore the eigenvector in question must be $\bfv_2$ and we have
both $\hat{\bfe}_1 \cdot \bfv_2 = 0$ and
$\hat{\bfe}_2 \cdot \bfv_2 = 0$.  Of course, it also follows from
the hypothesis $\hat{\bfU} \ne \bfU$ that $\hat{\bfe}_1 $ is not parallel
to either $\bfv_1$ or $\bfv_3$.  The remaining condition of the
cofactor conditions is the inequality (\ref{cc3}).  Clearly, these
necessary conditions are also sufficient for the cofactor conditions.
\end{proof}
\begin{figure}[htbp]
\begin{center}
\subfigure{\includegraphics[width = 0.7 \textwidth]{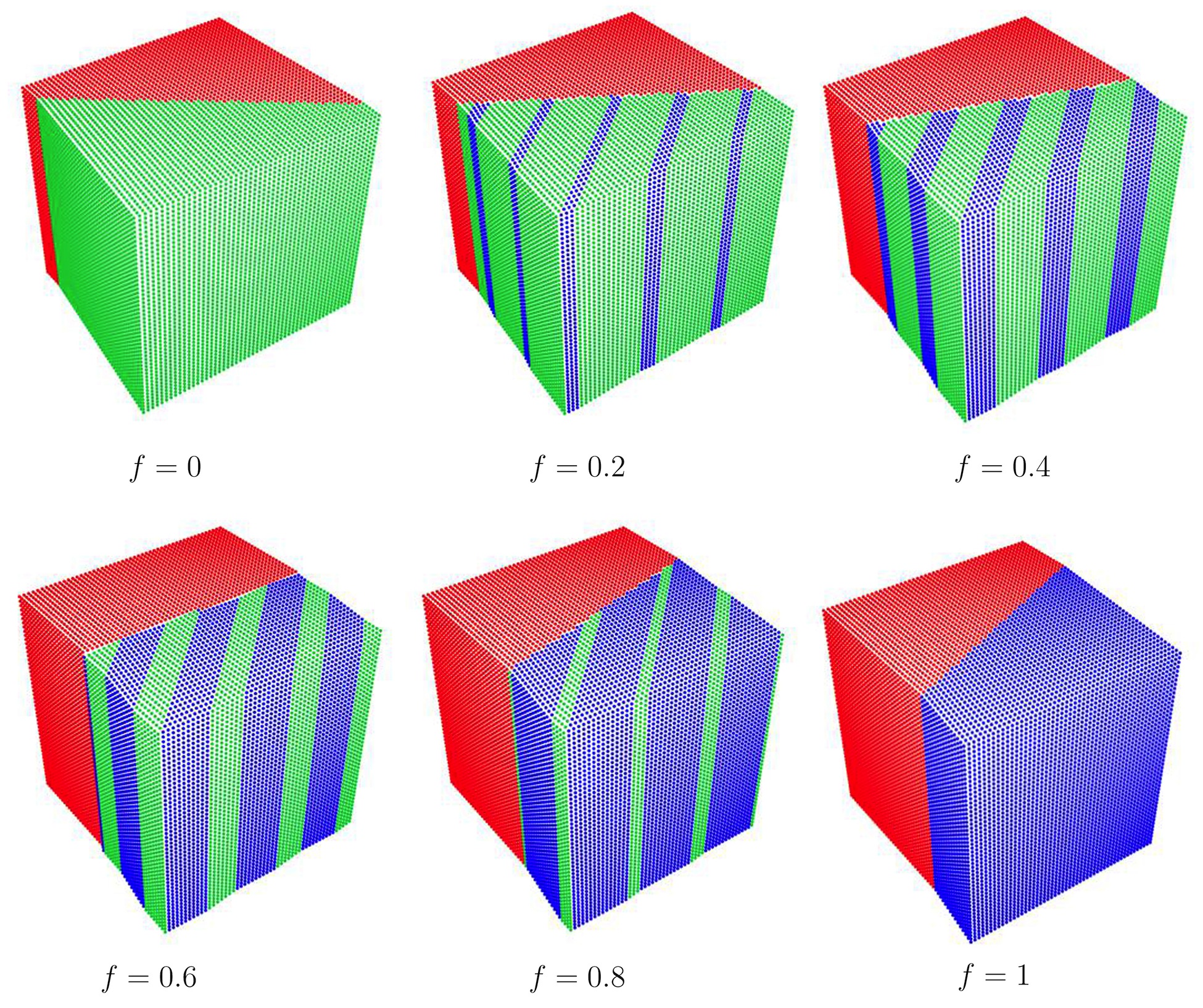}}\hspace{20pt}
%\subfigure{\includegraphics[width = 0.8 \textwidth]{comp_lines.jpg}}\hspace{20pt}
\caption{Austenite/martensite interfaces for Compound
twin system satisfying the cofactor conditions at various $f$ from 0 to 1.
The deformation is a plane strain.  In this case there is an elastic distortion near the habit plane. }
\label{comptwin}
\end{center}
\end{figure}

This result says that we satisfy cofactor conditions for Compound domains by
putting the orthonormal vectors  $\hat{\bfe}_1, \hat{\bfe}_2$ in the
$\bfv_1, \bfv_3$ plane and satisfying the inequality (\ref{cc3}).
If $\bfU$ is given as above, there is then only one degree-of-freedom,
say, the angle $\theta$ between $\hat{\bfe}_1$ and $\bfv_1$, in the
assignment of $\hat{\bfe}_1, \hat{\bfe}_2$. The left hand side of the
inequality (\ref{cc3}) then becomes a function of $\lambda_1, \lambda_3$
and $\theta$.  Given $\theta$, it can be seen from numerical examples
that there is a domain in $\rz^2$ of possible values of
$\lambda_1, \lambda_3$ at which (\ref{cc3}), and therefore the
cofactor conditions, are satisfied.
For many choices of $\theta$ this domain seems to be quite large,
including many potential alloys, but does not include all of
$\lambda_1 < 1 < \lambda_3$.  We do not see any general statements one
can make about this domain, except the obvious point that
if $\theta$ is fixed, then the left hand side of the inequality
(\ref{cc3}) tends to 0 as $|\lambda_3-1| + |1-\lambda_1| \to 0$.

It should be noted that except for the possibility of a restricted
domain for $\lambda_1, \lambda_3$, Compound domains can satisfy the
cofactor conditions {\it merely by symmetry and $\lambda_2 = 1$}.
That is, if the lattice parameters
of a potential alloy are first tuned to satisfy $\lambda_2 = 1$, and
the symmetry happens to be such that there are two 180 degree rotations in
the point group $\calP$ with perpendicular axes that lie in
a plane perpendicular to $\bfv_2$, then the cofactor conditions are
satisfied as long as the domain for $\lambda_1, \lambda_3$ is
suitable.     See the example of VO$_2$ in
Section \ref{alloys}.

There seem to be no general statements about the elimination of
the transition layer that one can make that are independent of the
choice of $\hat{\bfe}_1$ (satisfying Lemma \ref{lemcomp}), as was
done in the cases of Types I and II domains.  For example, the main
condition $\bfR_0 = \bfR_1$ that eliminated the transition layer for
Type I domains becomes a single scalar equation restricting
$\lambda_1, \lambda_3$ and $\theta$ in the case of Compound domains.
It may well be possible
for quite special choices of $\lambda_1, \lambda_3$ and $\theta$ to
eliminate the transition layer.   For practical alloy development
such a condition seems not so useful, as usually $\theta$ would
be given, and the resulting further restriction on $\lambda_1, \lambda_3$
would seem to be difficult to satisfy.  But further investigation
 is warranted.

\section{Simultaneous satisfaction of the cofactor conditions for different domain systems} \label{satisfaction_cofactor}

In the introduction we have argued that the cofactor conditions imply the existence of
many deformations consistent with the coexistence of austenite and martensite, and many
of these cases also have zero elastic energy.  Here we quantify these statements for one of
the two types of cubic to monoclinic phase transformations (see, e.g., \citet{soligo_99} and
\citet{james_2000}). 
This case is interesting with regard to applications (see Section \ref{alloys}), and is
representative of other high-to-low symmetry cases.

We consider symmetry change from cubic to monoclinic with $<\!\!100\!\!>_\textup{a}$ as the inherited 
2-fold axis. There are 12 martensite variants in this case
with transformation stretch matrices given by
\beq
\begin{array}{cccc}
\mathbf U_1 = \begin{bmatrix}\alpha&\beta&0\\\beta&\delta&0\\0&0&\gamma\end{bmatrix}, &\mathbf U_2 = \begin{bmatrix}\alpha&-\beta&0\\-\beta&\delta&0\\0&0&\gamma\end{bmatrix},&\mathbf U_3 = \begin{bmatrix}\delta&\beta&0\\\beta&\alpha&0\\0&0&\gamma\end{bmatrix},&\mathbf U_4 = \begin{bmatrix}\delta&-\beta&0\\-\beta&\alpha&0\\0&0&\gamma\end{bmatrix},\\ \\
\mathbf U_5 = \begin{bmatrix}\gamma&0&0\\0&\delta&\beta\\0&\beta&\alpha\end{bmatrix}, &\mathbf U_6 = \begin{bmatrix}\gamma&0&0\\0&\delta&-\beta\\0&-\beta&\alpha\end{bmatrix},&\mathbf U_7 = \begin{bmatrix}\alpha&0&\beta\\0&\gamma&0\\\beta&0&\delta\end{bmatrix},&\mathbf U_8 = \begin{bmatrix}\alpha&0&-\beta\\0&\gamma&0\\-\beta&0&\delta\end{bmatrix},\\ \\
\mathbf U_9 = \begin{bmatrix}\delta&0&\beta\\0&\gamma&0\\\beta&0&\alpha\end{bmatrix}, &\mathbf U_{10} = \begin{bmatrix}\delta&0&-\beta\\0&\gamma&0\\-\beta&0&\alpha\end{bmatrix},&\mathbf U_{11} = \begin{bmatrix}\gamma&0&0\\0&\alpha&\beta\\0&\beta&\delta\end{bmatrix},&\mathbf U_{12} = \begin{bmatrix}\gamma&0&0\\0&\alpha&-\beta\\0&-\beta&\delta\end{bmatrix}.\\ 
\end{array} \label{variants:monoI}
\eeq

To avoid degeneracies, we assume for the rest of this section that $\alpha \ne \delta$
and that the eigenvalues of $\bfU_1$ are distinct.  Between these martensite variants, 
there are 24 Type I twins, 24 Type II twins, 24 Compound twins, 24 Type I
domains, 24 Type II domains and 12 Compound domains.  
These domains with labels of pairs of compatible variants are listed in Table \ref{twintbl:monoI}.
The notation for variants is consistent with \eqref{variants:monoI}. 

In the case of domains that are not conventional twins (Table \ref{twintbl:monoI}), 
the rotation relating each
pair of compatible variants is a 90$^\circ$ rotation.  The 180$^\circ$ rotation that necessarily relates
these variants is given by formulas in the appendix.

The colored boxes in Table \ref{twintbl:monoI} have the property that if one pair
in the box satisfies the cofactor conditions for a certain type of domain, 
then all pairs in the box satisfy the
cofactor conditions for the same type of domain.  For example, if variants
1 and 6 have a Type I twin satisfying the cofactor conditions, then the Type I
twin pairs (2,5), (1,5) and (2,6) also satisfy the cofactor conditions.  In each of these cases
there are compatible triple junctions leading to numerous zero elastic energy microstructures
of austenite coexisting with martensite as discussed in Theorem \ref{thmI}.

The green box is particularly interesting.  If $\gamma=1$ (only) then the cofactor conditions
are satisfied (Lemma \ref{lemcomp}).  As can be seen from Table \ref{twintbl:monoI} there are then a very large number of Compound
domains that satisfy the cofactor conditions.  For each of these there are infinitely many
deformation gradients of martensite that coexist with $\bfI$ in the sense of the crystallographic
theory.  Thus, there is a huge collection of compatible deformations of austenite and martensite,
although none of these have zero elastic energy.  Under our hypotheses, Compound twins with
$\gamma \ne 1$ cannot satisfy the cofactor conditions, and the numerical evidence suggests that
this is also true for the Compound domains.

\definecolor{milky}{rgb}{1,1,0.85}
\definecolor{orange}{rgb}{0.9,0.5,0.2}
\definecolor{lightorange}{rgb}{1, 0.8, 0.1}
\definecolor{violet}{rgb}{1, 0.5, 0.6}
\renewcommand{\arraystretch}{1.2}
\newlength\Origarrayrulewidth
\newcommand{\Cline}[2]{%	
  \noalign{\global\setlength\Origarrayrulewidth{\arrayrulewidth}}%
  \noalign{\global\setlength\arrayrulewidth{2pt}}\arrayrulecolor{#1}\cline{#2}\arrayrulecolor{black}%
  \noalign{\global\setlength\arrayrulewidth{\Origarrayrulewidth}}%
}
\newcommand{\thickcell}[2]{%
  \multicolumn{1}{!{\color{#1}{\vrule width 2pt}}c!{\color{#1}{\vrule width 2pt}}}{#2}%
}
\begin{threeparttable}[htp]
\small
\centering
\caption{List of all possible twin systems for cubic to monoclinic transformations with $<\!\!100\!\!>_\textup{a}$ as the inherited 
2-fold axis. The notation $(i, j)$ presents domains which are symmetry related by $\mathbf U_i = \mathbf R \mathbf U_j \mathbf R^\textup{T}$, where $\mathbf R \in \mathcal P$ is characterized by the angle and rotational axis.  See text.}\label{twintbl:monoI}\vspace*{5pt}
    
\begin{tabular}{c|c|ll|c|c}

\hline

\multirow{2}{*}{Type} & \multirow{2}{*}{$\begin{array}{c}\mathbf R\\\theta^\circ, [h,k,l]\end{array}$} & \multicolumn{2}{c|}{\multirow{2}{*}{Type I/II twins or domains}} & \multicolumn{2}{c}{Compound twins or domains}\\
\cline{5-6}
\ &\ &\ & \ &$\gamma = 1$&$\gamma \neq 1$\\
 \hline
\Cline{green}{5-5}
    
\multirow{9}{*}{$\substack{\text{\normalsize Conventional}\\\text{\normalsize twins}}$}

& $180^\circ, [1, 0, 0]$ & & &\thickcell{green}{\renewcommand{\arraystretch}{0.8} $\begin{array}{c}(1, 2), (7, 8)\\(3, 4), (9, 10)\end{array}$\renewcommand{\arraystretch}{1.2}}  & \renewcommand{\arraystretch}{0.8} $\begin{array}{c}(1, 2), (7, 8)\\(3, 4), (9, 10)\end{array}$\renewcommand{\arraystretch}{1.2} \\  
\cline{2-2}

& $180^\circ, [0, 1, 0]$ & & & \thickcell{green}{ \renewcommand{\arraystretch}{0.8}$\begin{array}{c}(1, 2), (5, 6)\\(3, 4), (11, 12)\end{array}$ \renewcommand{\arraystretch}{1.2}} &  \renewcommand{\arraystretch}{0.8}$\begin{array}{c}(1, 2), (5, 6)\\(3, 4), (11, 12)\end{array}$ \renewcommand{\arraystretch}{1.2}\\  \cline{2-2}

& $180^\circ, [0, 0, 1]$ & & &\thickcell{green}{ \renewcommand{\arraystretch}{0.8}$\begin{array}{c}(7, 8), (11, 12)\\(5, 6), (9, 10)\end{array}$ \renewcommand{\arraystretch}{1.2}} &  \renewcommand{\arraystretch}{0.8}$\begin{array}{c}(7, 8), (11, 12)\\(5, 6), (9, 10)\end{array}$ \renewcommand{\arraystretch}{1.2} \\

 \cline{2-2} \Cline{blue}{3-3} \Cline{blue}{4-4}

& $180^\circ, [1, 0, 1]$& \thickcell{blue}{$(1, 6), (2, 5),$} & \thickcell{blue}{$(3, 12), (4, 11)$} & \thickcell{green}{$(7, 9), (8, 10)$} & $(7, 9), (8, 10)$ \\

\cline{2-2}

& $180^\circ, [1, 0, \bar{1}]$& \thickcell{blue}{$(1, 5), (2, 6),$} & \thickcell{blue}{$(3, 11), (4, 12)$} & \thickcell{green}{$(7, 9), (8, 10)$} & $(7, 9), (8, 10)$\\

\cline{2-2}\Cline{blue}{3-3} \Cline{blue}{4-4}

& $180^\circ, [1, 1, 0]$ &\thickcell{blue}{$(5, 10), (6, 9),$} & \thickcell{blue}{$(7, 12), (8, 11)$} & \thickcell{green}{$(1, 3), (2, 4)$} & $(1, 3), (2, 4)$ \\

\cline{2-2}

& $180^\circ, [\bar{1}, 1, 0]$ &\thickcell{blue}{$(5, 9), (6, 10),$} &\thickcell{blue}{$(7, 11), (8, 12)$} & \thickcell{green}{$(1, 3), (2, 4)$} & $(1, 3), (2, 4)$ \\

\cline{2-2}  \Cline{blue}{3-3}  \Cline{blue}{4-4}

& $180^\circ, [0, 1, 1]$& \thickcell{blue}{$(1, 8), (2, 7),$} & \thickcell{blue}{$(3, 10), (4, 9)$} & \thickcell{green}{$(5, 11), (6, 12)$} & $(5, 11), (6, 12)$ \\

\cline{2-2}

& $180^\circ, [0, \bar{1}, 1]$ & \thickcell{blue}{$(1, 7), (2, 8),$} & \thickcell{blue}{$ (3, 9), (4, 10)$} & \thickcell{green}{$(5, 11), (6, 12)$} & $(5, 11), (6, 12)$ \\  

  \Cline{blue}{3-3}  \Cline{blue}{4-4}  
\cline{1-4}

\multirow{6}{*}{$\substack{\text{\normalsize Domains} \\\ \text{(\normalsize all are} \\\ \text{\normalsize nonconventional}\\\text{\normalsize twins)}}$}

& \multirow{2}{*}{$90^\circ, [0, 0, 1]$}& \thickcell{blue}{$(5, 9), (6, 10),$} & \thickcell{blue}{$(7, 12), (8, 11)$} & \thickcell{green}{\multirow{2}{*}{$(1, 4), (2, 3)$}} & \multirow{2}{*}{$(1, 4), (2, 3)$} \\

%\cline{3-4}

& & \thickcell{blue}{$(5, 10), (6, 9),$} & \thickcell{blue}{$(7, 11), (8, 12)$} & \thickcell{green}{\ }&\\ 

 \Cline{blue}{3-3}  \Cline{blue}{4-4}   
\cline{2-4}

& \multirow{2}{*}{$90^\circ, [0, 1, 0]$}& \thickcell{blue}{$(1, 5), (2, 6),$} & \thickcell{blue}{$(3, 11), (4, 12)$} & \thickcell{green}{\multirow{2}{*}{$(7, 10), (8, 9)$}} & \multirow{2}{*}{$(7, 10), (8, 9)$} \\

%\cline{3-4}

& & \thickcell{blue}{$(1, 6), (2, 5),$} & \thickcell{blue}{$(4, 11), (3, 12)$} & \thickcell{green}{\ }& \\

\Cline{blue}{3-3}  \Cline{blue}{4-4}    
\cline{2-4}

& \multirow{2}{*}{$90^\circ, [1, 0, 0]$}& \thickcell{blue}{$(1, 8), (2, 7),$} & \thickcell{blue}{$(3, 10), (4, 9)$} & \thickcell{green}{\multirow{2}{*}{$(5, 12), (6, 11)$}} & \multirow{2}{*}{$(5, 12), (6, 11)$} \\

%\cline{3-4}

& & \thickcell{blue}{$(1, 7), (2, 8), $} & \thickcell{blue}{$(3, 9), (4, 10)$} & \thickcell{green}{\ }& \\
\Cline{blue}{3-3}
\Cline{blue}{4-4}
\Cline{green}{5-5}
\hline

\end{tabular} 
%\begin{tablenotes}
%\item[a] It is the condition such that the transformation stretch tensor has middle eigenvalue 1 and the cofactor conditions spontaneous %satisfied due to  Lemma \ref{lemcomp}.
%\end{tablenotes}
    
\end{threeparttable}

\section{Nucleation under the cofactor conditions}

The analysis given above suggests simple microstructures
with zero elastic energy that allow a continuous increase of the
volume of the new phase, starting at zero volume, in a material
satisfying the cofactor conditions. In a single
crystal there are obviously cases in which a layer of
martensite can grow in austenite and {\it vice versa}, merely
due to the condition $\lambda_2 =1$.
We illustrate some cases in which the set on which nucleation
takes place is lower dimensional, e.g., a line. As illustrated
and analyzed by \citet{ball_2011, ball_2011a} and \citet{seiner_2009},
%(In G. B. Olson, D. S. Lieberman and A. Saxena (eds.),
%Proceedings of the International Conference on Martensitic
%Transformations ICOMAT 2008, TMS.   Seiner, H., Landa, M.,
%Non-classical austenite-martensite interfaces observed
%in single crystals of Cu-Al-Ni. Phase Transitions 82 (2009),
%pp. 793-807.)
 the geometry of these nuclei are important for
nucleation phenomena.

\begin{figure}[htp]
\begin{center}
\subfigure{\includegraphics[width = 0.8 \textwidth]{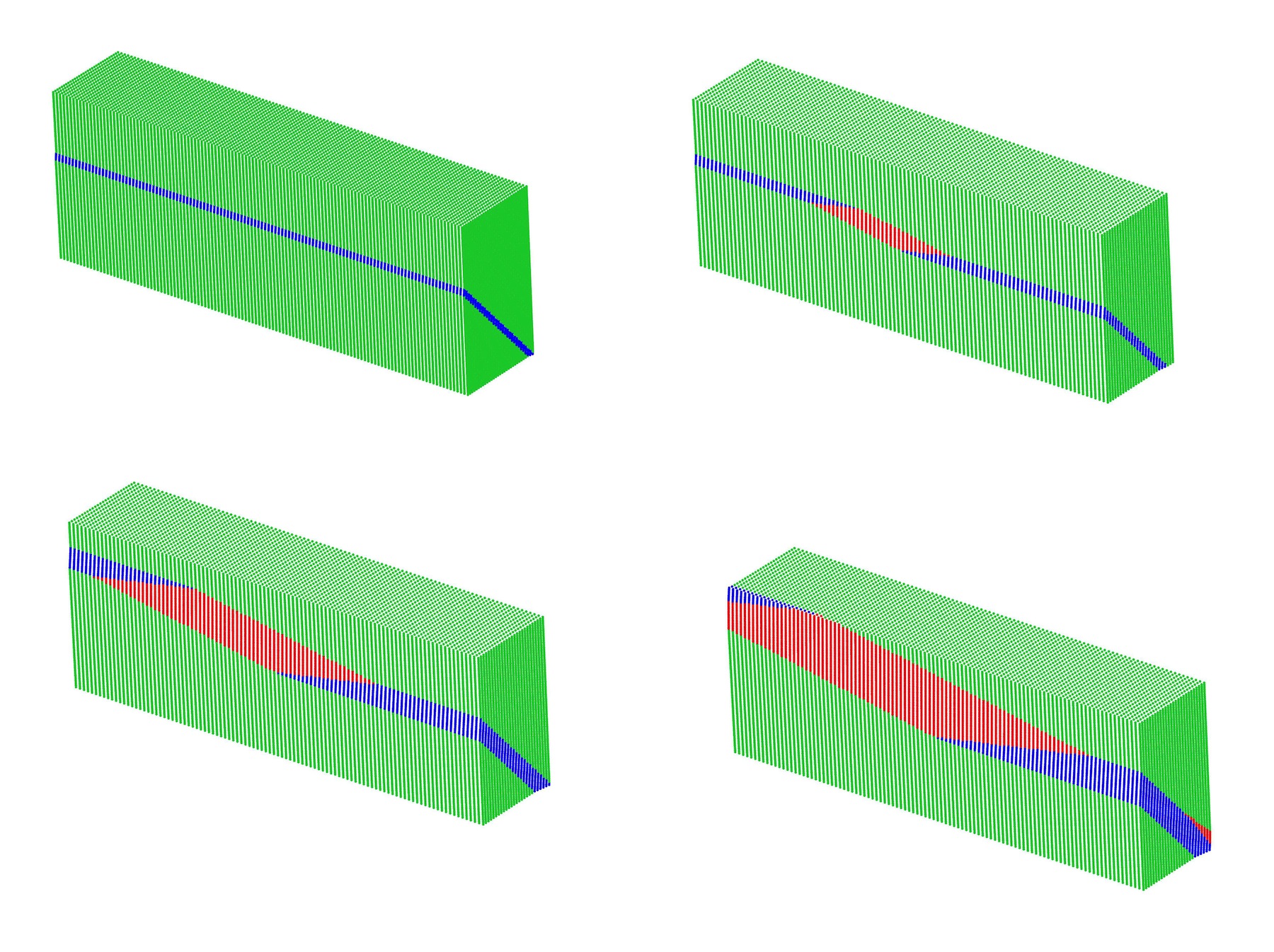}}\hspace{20pt}
\caption{ Example of nucleation of austenite (red) in a band of martensite
with zero elastic energy, under the cofactor condtions for Type I domains.
The blue and green are  two compatible variants
of martensite that can form a triple junction with austenite, as described
by Theorem \ref{thmI}. \label{nucleation_AinM}}
\end{center}
\end{figure}

An example of nucleation of austenite in martensite is given in
Figure \ref{nucleation_AinM}.  It is constructed from any Type I
domain for which the cofactor conditions are satisfied, and it uses
the  three deformation gradients $\bfI, \bfR_0^{\sigma} \bfU,
\bfR_0^{\sigma}\hat{\bfR} \hat{\bfU}$  given in Theorem  \ref{thmI}.
The regions on which these deformation gradients occur are shown as
red, green and blue, respectively,  in Figure \ref{nucleation_AinM}.
Nucleation in this case occurs on a line; four triple junctions
are simultaneously emitted from this line.

Under the same conditions, a simple mechanism for boundary nucleation
of martensite in austenite is shown in Figure \ref{nucleation_MinA}.  This is seen as a
simplified version of Figure \ref{curve_inf}.
\begin{figure}[htp]
\begin{center}
\subfigure{\includegraphics[width = 0.75 \textwidth]{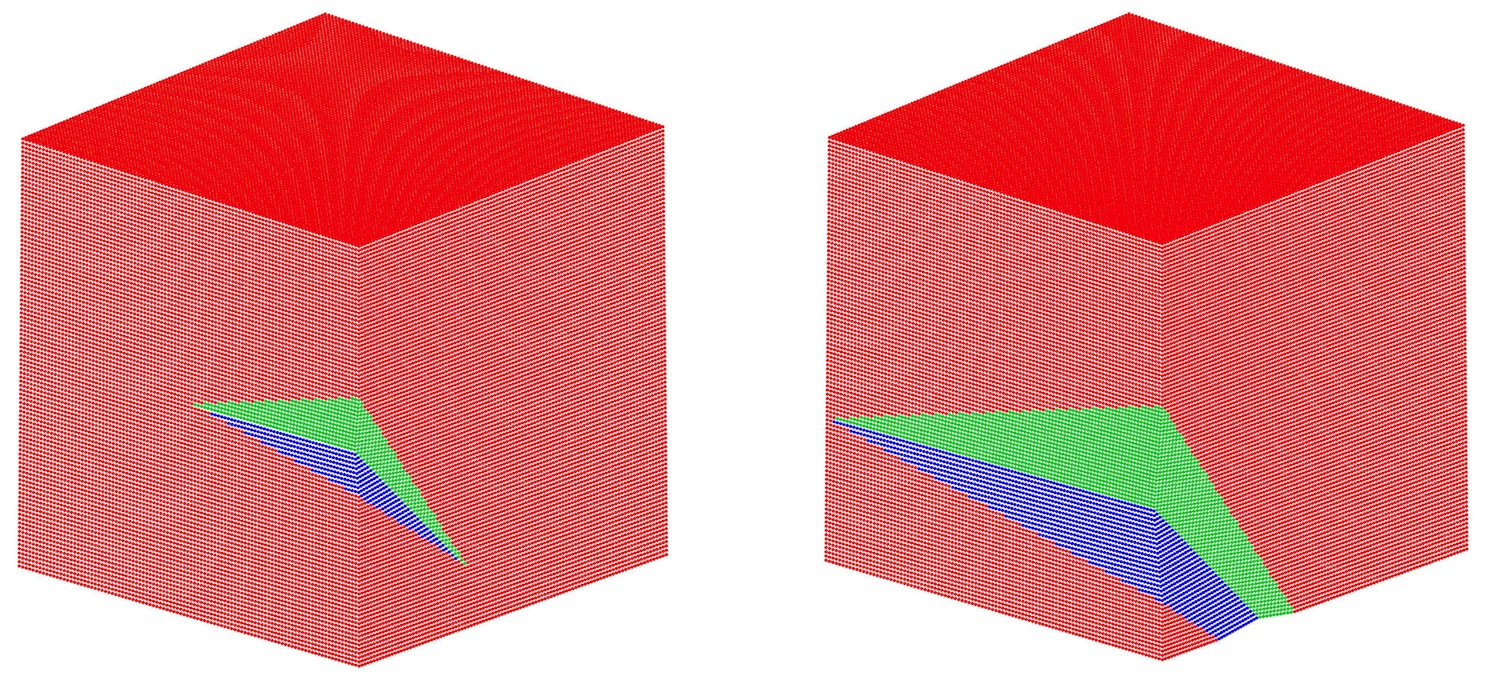}}\hspace{20pt}
\caption{ Example of nucleation of martensite (blue/green bands) in austenite (red lattice)
with zero elastic energy, with satisfaction of the cofactor conditions for Type I domains. \label{nucleation_MinA}}
\end{center}
\end{figure}

\section{Cofactor conditions in the geometrically linear case}

A number of versions of the geometrically linear theory of martensite
are in wide use for both fundamental theoretical and computational
studies \citep{khacha_1969, roitburd_1978, barsch_1984, kohn_1989, bhatt_1993,knupfer_2011}.  
There is a version of the cofactor conditions in the geometrically linear case.
Since the satisfaction of the cofactor conditions is expected to
have a dramatic effect on predicted microstructure and behavior in
the geometrically linear theory, we give these conditions here.

The cofactor conditions in geometrically linear theory
are different from the cofactor conditions in the geometrically
nonlinear theory, owing to the fact that the geometrically linear
theory is obtained from the geometrically nonlinear theory by
Taylor expansion \citep{bhatt_1993} or asymptotic analysis \citep{shimidt_08}.
As discussed below,
the cofactor conditions in the geometrically linear case should not
be used for alloy development in materials with appreciable
transformation strain.

The cofactor conditions in the geometrically linear case can be obtained
in two ways: i) by formal linearization of the cofactor conditions in
the geometrically nonlinear case following the expansion
given in \citet{ball_1992}, or ii) by writing down the equations of the
crystallographic theory of martensite in the geometrically linear case,
and imposing the condition that they be satisfied for any volume
fraction $0 \le f \le 1$.  The latter method is preferable because it
proves the existence of actual energy minimizing microstructures (or minimizing
sequences) for a broad family of geometrically linear theories of martensite.
We therefore follow method (ii).

The geometrically linear version of the crystallographic theory of
martensite in the cubic-to-tetragonal case first appeared in a paper of \citet{barkurt_1953} in the same
issue of AIME Journal of Metals as the general version of the crystallographic
theory by \citet{wechsler_1953}.
%Wechsler, Lieberman and Read

The basic kinematics of geometrically linear theory is the same as linearized
elasticity: it is based on the
displacement gradient $\nabla \bfu = \bfH \in \mathbb R^{3 \times 3}$,
which is decomposed into symmetric and skew parts $\bfH = \bfS + \bfW$,
$\bfS = \bfS^T, \ \bfW = -\bfW^T$ representing infinitesimal strain and
rotation.  A particular strain $\bfS = \bfE$ is given as the {\it transformation
strain}, and strains associated with the variants of martensite
are obtained by symmetry.  As above, we consider another variant
defined by the strain $\hat{\bfE}  = \bfQ \bfE \bfQ^T$ where
$\bfQ = -\bfI + 2 \hat{\bfe} \otimes \hat{\bfe}$, $|\hat{\bfe}| = 1$.
The basic compatibility
condition for variants with displacement gradients
$\nabla \bfu = \hat{\bfE} + \hat{\bfW}$ and $\nabla \bfu = \bfE$ is
\beq
 \hat{\bfE} + \hat{\bfW} -  \bfE = \bfa \otimes \bfn.   \label{lin_twin}
\eeq
(This is also the direct geometric linearization of (\ref{twin_rel2}).)
Taking the symmetric part of (\ref{lin_twin}) we have the compatibility
condition of geometrically linear theory:
\beq
 \hat{\bfE}  -  \bfE = \frac{1}{2}(\bfa \otimes \bfn + \bfn \otimes \bfa).
 \label{lin_twin2}
\eeq
By taking the trace, we have necessarily that $\bfa \cdot \bfn = 0$.
The basic lemma governing solutions of (\ref{lin_twin2}) is the following.
\begin{lem} \label{lin_comp} Necessary and sufficient conditions that
$\bfS \in \mathbb R^{3 \times 3}_\text{sym}$  is expressible in the form
$\bfS  = (1/2) (\bfa \otimes \bfn + \bfn \otimes \bfa)$ for some nonzero
$\bfa, \bfn \in \mathbb R^3$ is that the middle eigenvalue of $\bfS$
is zero.  If
$\bfS = s_1 \bfe_1 \otimes \bfe_1 + s_3 \bfe_3 \otimes \bfe_3$
with $\bfe_1, \bfe_3$ orthonormal and $ s_1 \le 0 \le s_3$, then
solutions $\bfa, \bfn$ of $\bfS  = (1/2) (\bfa \otimes \bfn + \bfn \otimes \bfa)$
can be taken as
\beq
\bfa = \sqrt{-s_1} \bfe_1 + \sqrt{s_3} \bfe_3, \quad
\bfn = -\sqrt{-s_1} \bfe_1 + \sqrt{s_3} \bfe_3.   \label{lin_an}
\eeq
These are unique up to switching $\bfa \to \bfn,\ \bfn \to \bfa$
and scaling $\bfa \to \mu \bfa, \
\bfn \to (1/\mu) \bfn$, $\mu \ne 0$.
\end{lem}
\begin{proof} (See e.g., \citet{kaushik_03})
Briefly, it is clear that a necessary condition that $\bfS$ has the given form
is  that $\bfS$ has an eigenvalue
equal to zero.  By examining the quadratic form
$\bfz \cdot \bfS \bfz$ with $\bfz$ taken as a bisector of $\bfa$ and $\bfn$,
and as a vector in the $\bfa, \bfn$ plane that is perpendicular to this
bisector, it is seen that the zero eigenvalue is the middle one.
The converse is proved by direct calculation using (\ref{lin_an}).
\end{proof}

In the special case that $\hat{\bfE} = \bfQ \bfE \bfQ^T$ as given above,
an alternative
representation of a solution of (\ref{lin_twin2}) is possible:
\beq
 \bfa = 4 \big( (\hat{\bfe} \cdot \bfE \hat{\bfe})\hat{\bfe} -
 \bfE \hat{\bfe} \big), \quad \bfn = \hat{\bfe}.    \label{alt_an}
\eeq
This form of the solution can be interpreted as the geometric linearization
of the Type I/II domains.  That is, due to the switching invariance
of Lemma \ref{lin_comp}, there exist infinitesimal rotations
$\hat{\bfW}_I  = -\hat{\bfW}_I^T$ and $\hat{\bfW}_{II} = -\hat{\bfW}_{II}^T$ such that, with
$\bfa$ and $\bfn$ defined by (\ref{alt_an}),
\beq
\hat{\bfE} + \hat{\bfW}_I   -  \bfE = \bfa \otimes \bfn, \quad  \quad
\hat{\bfE} + \hat{\bfW}_{II}   -  \bfE = \bfn \otimes \bfa,
\eeq
i.e., either $\bfa$ or $\bfn$ can be considered the interface normal.
$ \hat{\bfW}_I = - \hat{\bfW}_{II}$ as defined by these formulas
is necessarily skew.

From these compatibility conditions and the comments of Section \ref{crythe}
it is seen
that the equations of the crystallographic theory of martensite in
the geometrically linear case are the following.   Given
$\bfE \in \mathbb R^{3 \times 3}_\text{sym}$ and $\hat{\bfE} = \bfQ \bfE \bfQ^T$
as above, so that $ \hat{\bfE}  -  \bfE = \frac{1}{2}(\bfa \otimes \bfn +
\bfn \otimes \bfa)$ for some $\bfa, \bfn \in \mathbb R^3$, find
$\bfb_f, \bfm_f \in \mathbb R^3$ and $0 \le f \le 1$ such that
\beq
   f \hat{\bfE} + (1-f) \bfE = \frac{1}{2}
   (\bfb_f \otimes \bfm_f + \bfm_f \otimes \bfb_f).  \label{lin_cry}
\eeq
The {\it cofactor conditions in geometrically linear theory} are
necessary and sufficient conditions that there
exist $\bfb_f, \bfm_f \in \mathbb R^3$ satisfying (\ref{lin_cry})
for every  $0 \le f \le 1$. An explicit form of these conditions is
given in the following theorem.

\begin{theorem} \label{lin_cof} (Cofactor conditions in the geometrically
linear theory)  Let $\bfE \in \mathbb R^{3 \times 3}_\text{sym}$ and
$\hat{\bfe} \in  \mathbb R^3,\ |\hat{\bfe}| = 1$, be given.  Define
$\hat{\bfE}  = \bfQ \bfE \bfQ^T$ where
$\bfQ = -\bfI + 2 \hat{\bfe} \otimes \hat{\bfe}$, suppose that
$\hat{\bfE} \ne \bfE$, and define $\bfa, \bfn$ by (\ref{alt_an}).  There
exist $\bfb_f, \bfm_f \in \mathbb R^3$ satisfying (\ref{lin_cry}) for
every $0 \le f \le 1$ if and only if
\begin{align*}
&\eps_2 = 0, \text{ where $\eps_2$ is the middle eigenvalue of
$\mathbf E$, and rank $\bfE = 2$,}\tag{CCL1}\label{ccl1}\\
&(\bfa \cdot \bfv_2)(\bfn \cdot \bfv_2) = 0, \  \text{where}\ \bfE\bfv_2 = 0,
\ |\bfv_2| = 1, \tag{CCL2}\label{ccl2}\\
& \big(\textup{tr}(\bfE + \hat{\bfE})\big) ^2 -
\textup{tr}\big((\bfE + \hat{\bfE})^2 \big) \le 0.  \tag{CCL3}\label{ccl3}
\end{align*}
\end{theorem}

\begin{proof} Necessity of the conditions (CCL).   Clearly $\eps_2 = 0$
is a necessary condition at $f = 0$.  Also, $\bfE$ cannot vanish because
$\hat{\bfE} \ne \bfE$.  Potentially, $\bfE$ could be of rank 1,
$\bfE = \bfg \otimes \bfg \ne 0$, but then we would have $\hat{\bfE} =
\hat{\bfg} \otimes \hat{\bfg}$ with $|\bfg| = |\hat{\bfg}|$
and $\bfg \nparallel \hat{\bfg}$.
The unique zero eigenspace of $f \hat{\bfE} + (1-f) \bfE$ for $0<f<1$
would then be the 1-D subspace $\delta\, \bfg \times \hat{\bfg}$,
$\delta \in \mathbb R$.  The only possibility that the corresponding
zero eigenvalue of $f \hat{\bfE} + (1-f) \bfE$
would be its middle eigenvalue is that it is a double eigenvalue, because
the quadratic form $f \bfz \cdot \hat{\bfE}\bfz + (1-f) \bfz \cdot \bfE \bfz$
is clearly positive semidefinite.  This contradicts that the zero eigenspace
is one dimensional.  Hence, rank $\bfE = 2$.

The necessity of (\ref{ccl2}) follows by direct calculation of the
determinant of $f\hat{\mathbf E}+(1-f)\mathbf E$.  That is, if we write $\bfE = diag(\eps_1, 0, \eps_3)$
for $\eps_1< 0 < \eps_3$ (using (\ref{ccl1})), a direct calculation gives
\beq
  \det \big( f \hat{\bfE} + (1-f) \bfE \big) = \det \big(\bfE + (f/2)(\bfa \otimes \bfn
  + \bfn \otimes \bfa)\big) =
  4 f(1-f) \eps_1 \eps_3 (\bfa \cdot \bfv_2)(\bfn \cdot \bfv_2).
\eeq
The remaining necessary condition is that the implied
zero eigenvalue is the middle one.  Assume  (\ref{ccl1})and (\ref{ccl2})
and let $\eps_1^f, 0, \eps_2^f$
be the eigenvlaues of $f \hat{\bfE} + (1-f) \bfE$, with no particular
ordering.  If 0 is the middle eigenvalue, then $\eps_1^f \eps_2^f \le 0$
for $0 \le f \le 1$.  The quantity $\eps_1^f \eps_2^f$ is the second invariant
of $f \hat{\bfE} + (1-f) \bfE$.  This invariant is quadratic in $f$ and
has the same values at $f = 0,1$, and so it has the form
${\rm II}_f = \alpha f(1-f) + \eps_1 \eps_3$.  The coefficient $\alpha$
can be evaluated from $\alpha = d{\rm II}_f(0)/df = -\bfa \cdot \bfE \bfn$.
Also, $\alpha \ge 0$ by
$\bfa \cdot \bfE \bfn = \bfE \cdot (\hat{\bfE} - \bfE)$ and the
Cauchy�-Schwarz inequality, $\hat{\bfE} \cdot \bfE \le
|\bfE||\hat{\bfE}| = |\bfE|^2  = \bfE \cdot \bfE$.  Therefore, the largest
value of $\eps_1^f \eps_2^f \le 0$  occurs at $f = 1/2$, and so we have
the necessary condition ${\rm II}_{1/2} \le 0$ which is (\ref{ccl3}).
The conditions (\ref{ccl1}), (\ref{ccl2}) and (\ref{ccl3})
are also sufficient for (\ref{lin_cry}) to be satisfied for
every $0 \le f \le 1$, since they imply that the middle eigenvalue
of $f \hat{\bfE} + (1-f) \bfE$ is zero for all $0 \le f \le 1$.
\end{proof}

The explicit form of the conditions (\ref{ccl1})-(\ref{ccl3}) in the
eigenbasis of $\bfE$ is
\begin{align*}
&\bfE = diag(\eps_1, 0, \eps_3), \ \ \eps_1 < 0 < \eps_3, \tag{CCL1'}
\label{ccla} \\
& n_2^2 (n_1^2 \eps_1 + n_3^2 \eps_3) = 0, \tag{CCL2'}
\label{cclb} \\
& \left\{ \begin{array}{ll}
 \eps_1 \eps_3 + n_1^2 n_3^2 (\eps_3-\eps_1)^2 \le 0, & {\rm if}\ n_2 = 0, \\
  \eps_1 \eps_3 + n_3^2 \eps_3(\eps_3-\eps_1) \le 0, & {\rm if}\ n_1^2 \eps_1 + n_3^2 \eps_3 = 0.
  \end{array}  \right. \tag{CCL3'}
\label{cclc}
\end{align*}

As expected, the elastic transition layer can also be eliminated in the geometrically
linear case.  This occurs if $n_1^2 \eps_1 + n_3^2 \eps_3 = 0$.
It follows from $n_1^2 \eps_1 + n_3^2 \eps_3 = 0$ and (\ref{ccla}), (\ref{cclb})
that $\bfb_0 \parallel \bfb_1$ or $\bfm_0 \parallel \bfm_1$,
which in turn lead to triple junctions or parallelism, analogous to the
nonlinear case.

As mentioned above, one should be cautious on  applying the cofactor
conditions of geometrically linear theory in alloy development
because of the errors of geometric
linearization.  As a particular example, we can consider the main condition
(\ref{cc2'}) in the case of Types I and II domains.  According to
Proposition \ref{propI/II}, the condition (\ref{cc2'}) is $|\bfU^{-1} \hat{\bfe}| = 1$
for Type I domains and $|\bfU \hat{\bfe}| = 1$ for Type II domains under
the general hypotheses given there.   Both of these conditions linearize
to the same condition $\hat{\bfe}\cdot \bfE \hat{\bfe} =
n_1^2 \eps_1 + n_3^2 \eps_3 = 0$ of
(\ref{cclb}) (Recall from (\ref{alt_an}) that $\bfn = \hat{\bfe}$).
If we use the standard way
of evaluating the transformation
strain of linearized theory, $\bfE = \bfU - \bfI$, we have
\beq
  \begin{array}{ll} {\rm Geometrically\ nonlinear,\ Type\ I}: &
   (\frac{1}{\lambda_1^2} -1)n_1^2 +
  (\frac{1}{\lambda_3^2} -1)n_3^2 = 0,  \vspace*{0.5mm} \\
  {\rm Geometrically\ nonlinear,\ Type\ II}: &
  (\lambda_1^2 -1)n_1^2 +
  (\lambda_3^2 -1)n_3^2 = 0,            \vspace*{0.5mm} \\
  {\rm Geometrically\ linear}: &  (\lambda_1 -1)n_1^2 +
  (\lambda_3 -1)n_3^2 = 0.    \end{array}  \label{gngl1}
\eeq
As a numerical example, we can take  typical twin systems in a
cubic to monoclinic case discussed in Section \ref{satisfaction_cofactor}, which is also represented by the particular alloys identified in Section \ref{alloys} as good starting points for alloy development. For example, we take $\bfn = \hat{\bfe} = (1,1,0)/\sqrt{2}$ (in the cubic basis).
We take a typical measured value of $\lambda_3 = 1.08$.
Then, the exact satisfaction of
the cofactor conditions in the three cases of (\ref{gngl1}) gives
\beq
  \begin{array}{ll} {\rm Geometrically\ nonlinear,\ Type\ I}: &
  \lambda_1 = 0.936 ,  \\
  {\rm Geometrically\ nonlinear,\ Type\ II}: &
  \lambda_1 =  0.913,             \\
  {\rm Geometrically\ linear}: &  \lambda_1 = 0.920 .    \end{array}  \label{gngl2}
\eeq
In light of the sensitive dependence of hysteresis on the middle eigenvalue
seen on the horizontal axis of Figure \ref{hires},
the discrepancies seen in (\ref{gngl2}) may be significant.
Of course, it is no more difficult to use the geometrically
nonlinear conditions. The present situation with regard to the linearization
of the cofactor conditions
is similar to a number of other special lattice parameter
relationships discussed
by \citet{bhatt_1993}. In geometrically linear theory the elastic energy
near the habit plane can also be eliminated in some cases.

\section{Implications of the results for alloy development}
\label{alloys}

Although the theory justifying and explaining the cofactor conditions is
intricate, the conditions themselves are  simple and easy to implement.  One first chooses a domain system, which is the choice of
a unit vector $\hat{\bfe}$ relating two variants as in (\ref{ehati}).
Then one calculates $\bfa$ and $\bfn$ from (\ref{typeI_twin}) or
(\ref{comp}), depending on whether the domain system is Type I/II
or Compound.  As explained in Section \ref{twins_domains},
this choice also covers the cases of non-conventional and
non-generic twins, thus the terminology ``domain'' throughout this paper.
From these choices one identifies whether the domain is Type I, Type II or
Compound.

A convenient form of the cofactor conditions for alloy
development is then (\ref{cc1}) and
(\ref{cc2'}) (as further simplified by Proposition \ref{propI/II}).
The inequality (\ref{cc3}) also  has to be checked.  Among the systems
identified below that are near to satisfying the cofactor conditions,
it seems that this inequality will be automatically satisfied.  A useful
alloy development procedure is by interpolation:
\begin{enumerate}
\item From x-ray measurements determine the transformation stretch
matrix $\bfU$ and unit vector $\hat{\bfe}$ relating two variants:
$\hat{\bfU} = \bfQ \bfU \bfQ^T$, $\bfQ = -\bfI + 2 \hat{\bfe}\otimes
\hat{\bfe}$.  See  \citet{chen_12a} for an algorithm that automates
this part.  Identify the type of domain.  Below, for definiteness, it is
assumed that we wish to find an alloy satisfying the cofactor conditions
for a Type I twin system.  $\bfU$ depends on composition,
and we assume there are
two compositional variables $x$ and $y$.
\item  Determine a one-parameter family of alloys satisfying
$\lambda_2 = 1$.  We have found the following procedure
to be useful.  For each
$x$, find and alloy with composition $(x,y_1)$ having  $\lambda_2 \gtrsim 1$
and another alloy $(x,y_2)$ having  $\lambda_2 \lesssim 1$.  Then interpolate
to find a family of alloys with composition $(x, y(x))$ with $\lambda_2 = 1$.
\item  Among alloys with composition $(x,y(x))$, find an alloy with
composition $(x_1,y(x_1))$ with $|\bfU^{-1} \hat{\bfe}| \gtrsim 1$
and another alloy with composition
$(x_2,y(x_2))$ satisfying $|\bfU^{-1} \hat{\bfe}| \lesssim 1$.  Then interpolate to
find an alloy with composition $(x^{\star}, y(x^{\star}))$ satisfying
$|\bfU^{-1} \hat{\bfe}| = 1$, where $x^{\star}$ is between $x_1$ and $x_2$. This alloy satisfies (\ref{cc1}) and
(\ref{cc2}).
\item Check that the inequality (\ref{cc3}) is satisfied for the alloy
$(x^{\star}, y(x^{\star}))$.
\end{enumerate}
This procedure relies on the lattice parameters changing smoothly with
composition, as in Vegard's law.  This is often the case in a
suitable domain.  It also relies on having good starting points.

\vspace{5mm}
\renewcommand{\arraystretch}{1.3}
\begin{threeparttable}[htp]
{\small
\centering
\caption{Potential starting points for an alloy development program whose goal is to satisfy the cofactor conditions. }
\begin{tabular}{c|c|c|c}
\hline
{Candidates}&Cu$_{69}$Al$_{24}$Mn$_7$\tnote{1}&Au$_{25}$Cu$_{30}$Zn$_{45}$\tnote{2} &VO$_2$\tnote{3}\\

\hline
\begin{tabular}{l}
Crystal structure \\
\quad Austenite \\
\quad Martensite
\end{tabular}
&  \begin{tabular}{l}
    \\
DO3 \\
6M
\end{tabular}
&  \begin{tabular}{l}
  \\
L2$_1$ \\
M18R
\end{tabular}
&  \begin{tabular}{l}
  \\
Rutile  \\
Rutile monocl.
\end{tabular} \\
\hline
\begin{tabular}{l}
Bravais lattice \\
\quad Austenite \\
\quad Martensite
\end{tabular}
&  \begin{tabular}{l}
  \\
FCC  \\
Primitive monocl.
\end{tabular}
&  \begin{tabular}{l}
  \\
FCC  \\
Primitive monocl.
\end{tabular}
&  \begin{tabular}{l}
  \\
Primitive tetragonal \\
Base-centered monocl.
\end{tabular}  \\
\hline 
\begin{tabular}{l} Transformation \\ stretch matrix $\bfU$
\end{tabular}
&\!\!{\renewcommand{\arraystretch}{1} $\begin{bmatrix}1.1098&0.0279&0\\0.0279&1.0062&0\\0&0&0.8989\end{bmatrix}$ \renewcommand{\arraystretch}{1.2}}\!\!&
\!\!{\renewcommand{\arraystretch}{1}$\begin{bmatrix}1.0508&0&0.0142\\0&0.9108&0\\0.0142&0&1.0059\end{bmatrix}$\renewcommand{\arraystretch}{1.2}}\!\!&
\!\!{\renewcommand{\arraystretch}{1}$\begin{bmatrix}1.0669&0&0.0421\\0&0.9939&0\\0.0421&0&0.9434\end{bmatrix}$\renewcommand{\arraystretch}{1.2}}\!\!\\
\hline
{$|\lambda_2-1|$} & 0.0008 & 0.0018 & 0.0061 \\
\hline
{180$^\circ$ axis $\hat{\mathbf e}$}&$[011]$ or $[01\bar{1}]$&$[10\bar{1}]$&$[001]$\\
\hline
\begin{tabular}{l}
Cofactor conditions \\
 \quad Type I, $|\bfU^{-1} \hat{\bfe}|-1$ \\
 \quad Type II, $|\bfU \hat{\bfe}|-1$ \\
 \quad Compound
\end{tabular}
&  \begin{tabular}{l}
      \\
0.0256 \\
0.0202 \\
       \\
\end{tabular}
&\begin{tabular}{l}
      \\
0.0263 \\
0.029 \\
      \\
\end{tabular}
&\begin{tabular}{l}
      \\
      \\
      \\
satisfied if $\lambda_2 = 1$
\end{tabular}  \\
\hline
Inequality \eqref{cc3} &  0.0016  &  0.0175  &  0.0144  \\
\hline
\end{tabular}
}
\begin{tablenotes}
{\small
\item[1] \citep{zhang_07}
\item[2] \citep{hiroshi_76} The lattice parameters of austenite, which are needed to calculate $\bfU$,
 were not measured by these authors, so we have supplied this measurement.
\item[3] \citep{mcwhan_70}
}
\end{tablenotes}
\end{threeparttable}

\renewcommand{\arraystretch}{1}

\section*{Acknowledgements}
This work was supported by the MURI projects FA9550-12-1-0458 (administered by AFOSR) and W911NF-07-1-0410 (administered by ARO). This research was also benefited from the support of NSF-PIRE Grant No. OISE-0967140. The experimental work presented here was also partly supported by the Institute on the Environment and CharFac at the University of Minnesota. 

\newpage
\appendix
%\noindent {\bf \Large Appendix}
\section{Twin domains}
\normalsize
\vspace{5mm}

Here it is proved that general solutions of the equation of compatibility
(\ref{twin_rel2}) between martensite variants are represented as
Type I or Type II domains. 

\begin{proposition} \label{domains} Let  $\bfA = \bfA^T$ and $\bfB = \bfB^T$
be $3 \times 3$ positive-definite matrices satisfying
$\bfB  = \bfR \bfA \bfR^T$
for some $\bfR \in$ {\rm O(3)}.  Suppose $\bfA$ and $\bfB$
are compatible in the
sense that there is a matrix $\bfQ \in$ {\rm SO(3)}
such that
\beq
 \bfQ \bfB - \bfA = \bfa \otimes \bfn,   \label{com}
\eeq
$\bfa, \bfn \in \rz^3$. Then there is
a unit vector $\hat{\bfe} \in \rz^3$
such that
\beq
   \bfB = (-\bfI + 2 \hat{\bfe} \otimes \hat{\bfe})\bfA (-\bfI + 2 \hat{\bfe} \otimes \hat{\bfe}).
   \label{180}
\eeq
Conversely, if $3 \times 3$ matrices $\bfA$ and $\bfB$ satisfy (\ref{180})
for some unit vector $\hat{\bfe}$, then
 there is $\bfQ \in$ {\rm SO(3)} so that
(\ref{com}) is satisfied.  A formula for $\hat{\bfe}$ can be given as follows.
Under the hypotheses, there is an orthonormal basis
$\{\bfe_1, \bfe_2, \bfe_3\}$ such that
\beq
\bfA^{-1}\bfB^2\bfA^{-1}=\mu_{1} \bfe_{1} \otimes \bfe_{1}+
\bfe_2 \otimes \bfe_2+ \mu_3 \bfe_3 \otimes \bfe_3,   \label{C}
\eeq
where $0 < \mu_1\le 1 \le \mu_3$ and the following identities
hold:
\beq
  \mu_1 \mu_3 = 1, \quad \bfe_1 \cdot \bfA^2 \bfe_1 =
  \mu_3\, \bfe_3 \cdot \bfA^2 \bfe_3,
   \quad (\bfe_2 \cdot \bfA^2 \bfe_1)^2 =
   \mu_3 (\bfe_2 \cdot \bfA^2 \bfe_3)^2.
\eeq
In the case $\mu_3 >1$ all unit vectors $\hat{\bfe}$ satisfying
(\ref{180}) are given by
\beq
  \hat{\bfe}  = \pm (\delta_1 \bfA \bfe_1 + \delta_3 \bfA \bfe_3),
  \label{defhate}
\eeq
where
\beq
\delta_1 =  \bigg(2(\bfe_1 \cdot \bfA^2 \bfe_1  +
s \sqrt{\mu_3}\,  \bfe_3 \cdot \bfA^2 \bfe_1)\bigg)^{-1/2} \quad {\rm and}
\quad  \delta_3 = s \sqrt{\mu_3}\, \delta_1.   \label{deldef}
\eeq
and $s \in \{\pm 1\}$ satisfies
$s \sqrt{\mu_3} (\bfe_2 \cdot \bfA^2 \bfe_3) =
-\bfe_2 \cdot \bfA^2 \bfe_1$.
In the case $\mu_3=1$ necessarily $\bfB = \bfA$ and $\hat{\bfe}$ can be
chosen as a normalized eigenvector of $\bfA$.

\end{proposition}
In words:  for stretch matrices
related by orthogonal similarity as we
have for variants of martensite,
necessary and sufficient conditions for compatibility are
that these matrices are related by a 180$^{\circ}$ rotation.

\vspace{2mm}

\proof Without loss of generality, by replacing $\bfR$
by $-\bfR$ if necessary,  we can assume
$\bfR \in\ $SO(3).  The condition (\ref{C}), which under the given hypotheses
is necessary and sufficient for
(\ref{com}), has been proved in  \citet{ball_james_87}, Prop. 4.
 We can assume without loss of generality that
 $0 < \mu_1 < 1 < \mu_3$.  That is, if, say, $\mu_3 =1$, then
 by taking $\det$ of  (\ref{C}) and using $\det \bfA = \det \bfB$
 we would get necessarily $\mu_1 = 1$.  This would lead to
 $\bfA^2 = \bfB^2$.  Then by taking the positive-definite square root,
 we would have $\bfA = \bfB$.  This, in turn, would imply that
 (\ref{180}) is satisfied, for example, with $\hat{\bfe}$ equal to an
 eigenvector of $\bfA$. Thus, below we will assume $\mu_1 < 1 < \mu_3$.

There are several identities satisfied by the quantities on the
right hand side of (\ref{C}).  These follow from the hypothesis that
$\bfB  = \bfR \bfA \bfR^T$ which implies that $\bfA$ and $\bfB$ have the
same eigenvalues and therefore $\det \bfA = \det \bfB$,
${\rm tr} \bfA^2 = {\rm tr} \bfB^2$ and ${\rm tr} \bfA^4 = {\rm tr} \bfB^4$.
These in turn yield the following necessary conditions:
\begin{enumerate}
\item $\det \bfA = \det \bfB \Longrightarrow \mu_1 \mu_3 = 1$.
Obvious by taking $\det$ of (\ref{C}).
\item ${\rm tr} \bfA^2 = {\rm tr} \bfB^2 \Longrightarrow
\bfe_1 \cdot \bfA^2 \bfe_1 = \mu_3 \bfe_3 \cdot \bfA^2 \bfe_3$. This
follows by subtracting the identity matrix from (\ref{C}) and then
pre and post multiplying by $\bfA$ to get
\beq
\bfB^2  - \bfA^2 =
(\mu_{1}- 1) \bfA \bfe_{1}\otimes \bfA \bfe_{1} +
(\mu_3 - 1)\bfA \bfe_3\otimes \bfA \bfe_3.
\eeq
Taking the trace and using $\mu_1 \mu_3 = 1$ and $\mu_3 \ne 1$,
we get $\bfe_1 \cdot \bfA^2 \bfe_1 = \mu_3 \bfe_3 \cdot \bfA^2 \bfe_3$.
\item ${\rm tr} \bfA^4 = {\rm tr} \bfB^4 \Longrightarrow
(\bfe_2 \cdot \bfA^2 \bfe_1)^2 = \mu_3  (\bfe_2 \cdot \bfA^2 \bfe_3)^2$.
This follows from (\ref{C}) by pre and post multiplying by $\bfA$ to
get $\bfB^2 = \mu_{1} \bfA  \bfe_{1}\otimes \bfA \bfe_{1}+
\bfA \bfe_2 \otimes \bfA \bfe_2+ \mu_3 \bfA \bfe_3 \otimes \bfA \bfe_3$
then squaring this to get $\bfB^4$.  Now write $\bfA^2 = \bfA\, \bfI\, \bfA =
\bfA (\bfe_{1} \otimes  \bfe_{1}+
\bfe_2 \otimes  \bfe_2 +   \bfe_3 \otimes  \bfe_3) \bfA $ and square this
to get $\bfA^4$.  Put ${\rm tr} \bfA^4 = {\rm tr} \bfB^4$ and simplify
using items 1 and 2 and $\mu_3 \ne 1$ to get the result.
\end{enumerate}
 
Substituting provisionally the expression (\ref{180})
for $\bf B$
into (\ref{C}), we get the necessary condition
\beq
\bfA^{-1}(-\bfI+2\hat{\bfe}\otimes\hat{\bfe})\bfA^2(-\bfI+2\hat{\bfe}
\otimes \hat{\bfe}) \bfA^{-1}=
\mu_{1} \bfe_{1}\otimes \bfe_{1}+ \bfe_2\otimes \bfe_2+\mu_3
\bfe_3\otimes \bfe_3.  \label{basic}
\eeq
Multiplying out the tensor products in (\ref{basic}) we derive
\beq
 -2\bfA \hat{\bfe} \otimes \bfA^{-1}\hat{\bfe}-2\bfA^{-1}\hat{\bfe}
 \otimes\bfA\hat{\bfe}+
 4(\hat{\bfe}\cdot\bfA^2\hat{\bfe})\bfA^{-1}\hat{\bfe}
 \otimes\bfA^{-1}\hat{\bfe} =
 (\mu_{1}-1) {\bfe}_{1} \otimes {\bfe}_{1}+(\mu_3-1)
 {\bfe}_3 \otimes {\bfe}_3.
 \label{eqn}
\eeq
To solve this equation, we try to find a unit vector $\hat{\bfe}$ of the form
\beq
\hat{\bfe} = \sigma_1 \bfA^{-1}\bfe_1 + \sigma_3  \bfA^{-1}\bfe_3 =
\delta_1 \bfA \bfe_1
+ \delta_3 \bfA \bfe_3.  \label{2planes}
\eeq
The condition $1 = \hat{\bfe} \cdot \hat{\bfe} =
\bfA \hat{\bfe} \cdot\bf A^{-1} \hat{\bfe}$
implies that
 \beq
 \sigma_1 \delta_1+ \sigma_3 \delta_3=1.  \label{sd}
\eeq
Substituting the expressions for $\bfA \hat{\bfe}$ and $\bfA^{-1}
\hat{\bfe}$ into the
equation (\ref{eqn}), we get,
\beqs
-2( \sigma_1 \bfe_1+\sigma_3 \bfe_3)\otimes(\delta_1 \bfe_1+
 \delta_3 \bfe_3)
- 2(\delta_1 \bfe_1+ \delta_3 \bfe_3)\otimes( \sigma_1 \bfe_1+
 \sigma_3 \bfe_3) \nonumber \\
+ 4(\sigma_1^2+ \sigma_3^2)(\delta_1 \bfe_1+ \delta_3 \bfe_3)
\otimes( \delta_1 \bfe_1+ \delta_3 \bfe_3)=(\mu_{1}-1)  \bfe_{1}
\otimes \bfe_{1}+(\mu_3-1) \bfe_3 \otimes \bfe_3.
\eeqs
Rearranging similar terms in the above equation results in the following:
\beqs
\lefteqn{\big(- 4 \delta_1 \sigma_1 + 4( \sigma_1^2+ \sigma_3^2)
 \delta_1^2\big) \bfe_1\otimes \bfe_1}  \nonumber \\
& &   +
\big(-2 \sigma_1 \delta_3-2 \delta_1 \sigma_3+4( \sigma_1^2+
\sigma_3^2) \delta_1\delta_3\big)
(\bfe_1\otimes \bfe_3+ \bfe_3\otimes \bfe_1) \nonumber \\
& &+\big(-4 \delta_3 \sigma_3+ 4( \sigma_1^2+ \sigma_3^2)
\delta_3^2\big) \bfe_3\otimes \bfe_3 =
(\mu_1 -1) \bfe_{1} \otimes \bfe_{1}+(\mu_3 - 1)
\bfe_3 \otimes \bfe_3.
\label{rearr}
\eeqs
Comparing the 13  terms on both sides of (\ref{rearr}) and using (\ref{sd}),
we get the following expression connecting
$\bf \sigma_1, \bf \sigma_3, \bf \delta_1$ and $\bf \delta_3$.
\beq
(\sigma_1 \delta_3-  \delta_1  \sigma_3)(1-2 \sigma_3 \delta_3)=0.
\eeq
The vanishing of the first factor,
$\sigma_1\delta_3- \delta_1\sigma_3 = 0$, leads to the
trivial case  $\mu_1 = \mu_3 = 1$ which has been excluded above.
The vanishing of the
second factor gives that $\sigma_3 \delta_3=\frac{1}{2}$ and then from
 (\ref{sd}), $ \sigma_1 \delta_1=\frac{1}{2}$. This
shows that none of the unknowns $\delta_1, \sigma_1, \delta_3,  \sigma_3$
vanish. Now  the $\bfe_1\otimes \bfe_1$ and $\bfe_3\otimes \bfe_3$
terms in equation (\ref{rearr}) give
\beqs
4 \sigma_3^2 \delta_1^2 &=& \mu_1 \quad \Longrightarrow \quad
 \frac{\delta_1^2}{\delta_3^2} = \mu_1,  \nonumber \\
 4 \sigma_1^2 \delta_3^2 &=& \mu_3 \quad \Longrightarrow \quad
 \frac{\delta_3^2}{\delta_1^2} = \mu_3.  \label{dl}
\eeqs
These equations are consistent with $\mu_1 \mu_3 = 1$, and we only
need to retain one
of them.  In summary, (\ref{basic}) is satisfied for a unit
vector $\hat{\bfe}$ of the form
(\ref{2planes}) if and only if $\sigma_1, \sigma_3, \delta_1, \delta_3$
satisfy
\beq
 \sigma_1 \delta_1=\frac{1}{2}, \quad \sigma_3 \delta_3=\frac{1}{2}, \quad
 \delta_3^2 = \mu_3 \delta_1^2.
\eeq
A useful way to write this solution is:
\beq
\delta_3 = s \sqrt{\mu_3} \delta_1, \quad \sigma_1 =
\frac{1}{2 \delta_1},
\quad \sigma_3 =  \frac{s}{2 \sqrt{\mu_3} \delta_1}, \quad s = \pm 1.
\label{choice}
\eeq
So far, $\delta_1 \ne 0$ and $s = \pm 1$ are free parameters.

Although we have solved (\ref{basic}) by the choice (\ref{choice}),
we have to be sure that these values of $\delta_1, \delta_3, \sigma_1,
\sigma_3$
satisfy (\ref{2planes}).  This is a vector equation in 3D and
 therefore is equivalent to the three equations one gets by
 dotting it with the three linearly independent vectors,
 $\bfA\bfe_1, \bfA\bfe_2, \bfA\bfe_3$.  This gives the three
 equations,
\beqs
 \sigma_1 &=& \delta_1 (\bfe_1 \cdot \bfA^2 \bfe_1) +
 \delta_3 (\bfe_3 \cdot \bfA^2 \bfe_1),   \nonumber \\
  \sigma_3 &=& \delta_1 (\bfe_1 \cdot \bfA^2 \bfe_3) +
 \delta_3 (\bfe_3 \cdot \bfA^2 \bfe_3),   \nonumber \\
 0 &=& \delta_1 (\bfe_2 \cdot \bfA^2 \bfe_1) +
 \delta_3 (\bfe_2 \cdot \bfA^2 \bfe_3).  \label{3eq}
\eeqs
If we square the last equation and use (\ref{choice}) and the nonvanishing
of $\delta_1$, we get
\beq
(\bfe_2 \cdot \bfA^2 \bfe_1)^2 = \mu_3  (\bfe_2 \cdot \bfA^2 \bfe_3)^2.
\label{id1}
\eeq
This is satisfied by virtue of Item 3 above. Hence, the square of the third
equation of (\ref{3eq}) is an identity.  So, we can satisfy the third of
(\ref{3eq}) by an appropriate choice of $s = \pm 1$ of (\ref{choice}).
In particular, there exists $s \in \{ \pm 1 \}$ satisfying
\beq
s \sqrt{\mu_3} (\bfe_2 \cdot \bfA^2 \bfe_3) = -\bfe_2 \cdot \bfA^2 \bfe_1.
\label{s}
\eeq
This uniquely determines $s$ unless it happens that $\bfe_2 \cdot \bfA^2 \bfe_3 = 0$,
in which case also $\bfe_2 \cdot \bfA^2 \bfe_1=0$ and $s$ can be either $\pm 1$.
Now we further note
that the first two equations in (\ref{3eq}) are not independent.  That is,
multiply the first of these by $\delta_1 \ne 0$ and the second by
$\delta_3 \ne 0$,
subtract the equations
and use the conditions $\sigma_3 \delta_3 = \sigma_1 \delta_1 =
\frac{1}{2}$. This leads to
\beq
\bfe_1 \cdot \bfA^2 \bfe_1 - \mu_3 \bfe_3 \cdot \bfA^2 \bfe_3 = 0.
\label{id2}
\eeq
This is automatically satisfied, by virtue of Item 2 above. Hence, there is
only one independent equation in (\ref{3eq}), that we can take to be the first
one:
\beq
\frac{1}{2 \delta_1} = \delta_1 (\bfe_1 \cdot \bfA^2 \bfe_1)  +
s \sqrt{\mu_3} \delta_1 (\bfe_3 \cdot \bfA^2 \bfe_1),  \label{ford}
\eeq
that is,
\beq
2 \delta_1^2 \big(  (\bfe_1 \cdot \bfA^2 \bfe_1)  +
s \sqrt{\mu_3} (\bfe_3 \cdot \bfA^2 \bfe_1)\big)
 = 1.  \label{del1}
\eeq

We claim that, under our hypotheses, (\ref{del1}) can always be solved for
$\delta_1 \ne 0$.   That is, by the positive definiteness of $\bfA^2$,
we have $\bfe_1 \cdot \bfA^2 \bfe_1 > 0,\ \bfe_3 \cdot \bfA^2 \bfe_3 > 0$,
$(\bfe_1 \cdot \bfA^2 \bfe_1)(\bfe_3 \cdot \bfA^2 \bfe_3)
> (\bfe_3 \cdot \bfA^2 \bfe_1)^2$.  Hence, eliminating $\sqrt{\mu_3}$
using (\ref{id2}) (see Item 2), we have
for either choice $s = \pm 1$,
\beq
(\bfe_1 \cdot \bfA^2 \bfe_1)  +
s \sqrt{\mu_3}  (\bfe_3 \cdot \bfA^2 \bfe_1)
= \sqrt{\frac{\bfe_1 \cdot \bfA^2 \bfe_1}{\bfe_3 \cdot \bfA^2 \bfe_3}}
\left( \sqrt{(\bfe_1 \cdot \bfA^2 \bfe_1)(\bfe_3 \cdot \bfA^2 \bfe_3)}
+ s \bfe_3 \cdot \bfA^2 \bfe_1 \right) > 0.  \label{co1}
\eeq
Hence, $\delta_1$ given by (\ref{deldef}) is well-defined.  Equations
(\ref{del1}) and (\ref{choice}) imply that the vector $\hat{\bfe}$ given
by (\ref{2planes}) is a unit vector and satisfies (\ref{basic})
and therefore (\ref{180}).

The sufficiency of the condition (\ref{180}) for compatibility is a
standard result, see \citet{kaushik_03} or (\ref{typeI_twin}) above.
The formula for $\hat{\bfe}$
follows from (\ref{2planes}), (\ref{choice}) and (\ref{del1}) above.   \qed
 
\begin{cor} \label{cd} (Compound domains)  Assume the hypotheses of
Proposition \ref{domains}.  There are two unit vectors
$\hat{\bfe}_{+} \nparallel \hat{\bfe}_{-}$ satisfying (\ref{180})
if and only if
\beq
\bfe_2 \cdot \bfA^2 \bfe_3 =
\bfe_2 \cdot \bfA^2 \bfe_1 = 0.  \label{deg}
\eeq
If (\ref{deg}) is satisfied and $\mu_3 >1$, there are precisely two
such nonparallel
unit vectors (up to a premultiplied $\pm$) that satisfy
(\ref{180}), and in fact these vectors are orthonormal,
$\hat{\bfe}_{+} \cdot\hat{\bfe}_{-} = 0$.  They are given by  the formulas
\beq
  \hat{\bfe}_{\sigma}  = \delta_1^{\sigma} \bfA \bfe_1 +
  \delta_3^{\sigma} \bfA \bfe_3,  \quad \sigma = \pm, \label{defhatepm}
\eeq
where
\beq
\delta_1^{\sigma} =  \bigg(2(\bfe_1 \cdot \bfA^2 \bfe_1  +
\sigma \sqrt{\mu_3}\,  \bfe_3 \cdot \bfA^2 \bfe_1)\bigg)^{-1/2}
\quad {\rm and}
\quad  \delta_3^{\sigma} = \sigma \sqrt{\mu_3}\, \delta_1^{\sigma},
\quad \sigma  = \pm.   \label{deldef1}
\eeq
In the case $\mu_3=1$ necessarily $\bfB = \bfA$ and the solutions
$\hat{\bfe}$ of (\ref{180}) consist of unit vectors in the
eigenspace of $\bfA$.
\end{cor}

\proof
The proof follows immediately from the statement
$s \sqrt{\mu_3} (\bfe_2 \cdot \bfA^2 \bfe_3) =
-\bfe_2 \cdot \bfA^2 \bfe_1$ of Proposition \ref{domains}, which does
not uniquely determine $s \in \{\pm 1\}$ if and only if (\ref{deg}) is
satisfied.  The fact that the two solutions $\bfe_{\pm1}$ are nonparallel
is seen from their forms (\ref{deldef}), and the fact that these are the
only possible solutions up to premultiplied $\pm$ follows from
Proposition \ref{domains}.  The orthonormality of
$\hat{\bfe}_{+}$ and $\hat{\bfe}_{-}$ follows by direct calculation
using (\ref{defhatepm}) and  (\ref{deldef1}).
\qed

 \bibliographystyle{apalike}
 \bibliography{reference}

\end{document}